\documentclass[journal]{IEEEtran}

\IEEEoverridecommandlockouts

\usepackage{textcomp}
\usepackage{xcolor}
\usepackage{graphicx}
\usepackage{epstopdf}
\usepackage{cite}
\usepackage{amsmath}
\usepackage{amssymb}
\usepackage{amsfonts}
\usepackage{array}
\usepackage{color}
\usepackage{algorithm,algorithmic}
\usepackage{booktabs}
\usepackage{multirow}
\usepackage{url}
\usepackage{amsthm}
\usepackage{subcaption}
\usepackage{bm,soul,enumitem}
\usepackage{dsfont}
\usepackage{placeins}

\newtheorem{proposition}{Proposition}
\newtheorem{theorem}{Theorem}

\newtheorem{lemma}{Lemma}
\newtheorem{remark}{Remark}

\newcommand{\C}{\mathbb{C}}
\newcommand{\R}{\mathbb{R}}

\def\BibTeX{{\rm B\kern-.05em{\sc i\kern-.025em b}\kern-.08em
    T\kern-.1667em\lower.7ex\hbox{E}\kern-.125emX}}

\begin{document}

\title{Sensing-Induced Embodied Communication \\ in the Near Field}

\author{Jingreng Lei and Yulin Shao%
\thanks{J. Lei and Y. Shao are with the Department of Electrical and Computer Engineering, The University of Hong Kong, Hong Kong S.A.R. (e-mail: leijr@eee.hku.hk; ylshao@hku.hk).
}
}

\maketitle

\begin{abstract}
Integrated sensing and communication is turning the cellular infrastructure into an active observer of the physical world. 
When such infrastructure interacts with embodied agents capable of deliberately changing their states and surroundings, the physical world itself can become a communication medium. 
This paper studies the fundamental communication limits of this sensing-induced embodied communication paradigm in the near field. 
We consider an agent that maps messages to the positions of a controllable scatterer within a bounded three-dimensional (3D) region, while a base station decodes the selected position from multi-snapshot monostatic sensing echoes. 
Near-field spherical wavefronts resolve both angle and range, expanding the embodied-symbol space from a 2D plane to a 3D volume. 
This gain, however, comes with a position-dependent and anisotropic reliability geometry. 
We characterize this geometry through the pairwise Bhattacharyya distance and derive a local ellipsoidal representation of the resulting 3D confusability regions, whose principal axes quantify directional sensing resolution. 
The ellipsoid further degenerates into the 2D transverse ellipse in the far-field limit, unifying the two regimes. 
We then formulate the finite-snapshot $\epsilon$-capacity and translate reliable codebook design into a 3D packing problem. 
A face-centered cubic construction provides an achievable rate, while a geometric converse yields a complementary upper bound. 
Numerical results validate the proposed geometry and demonstrate the capacity gain of near-field volumetric packing over far-field planar packing. These results establish a unified geometric and information-theoretic framework for communication through deliberately configured physical states.
\end{abstract}


\section{Introduction}
\label{sec:introduction}

Wireless systems are being reshaped by two converging developments. 
On the infrastructure side, integrated sensing and communication (ISAC) \cite{TR22.837} is equipping cellular infrastructure with the ability to continuously probe and perceive the physical environment~\cite{liuSensingWithCommunicationSignals2026,gonzalezPrelcicTheIntegratedSensing2024,zhang2024optical,zhangAnOverviewOfSignalProcessing2021,chafiiTwelveScientificChallenges2023}. 
On the terminal side, embodied AI \cite{savva2019habitat,cuiLLMindOrchestratingAI2025} is expanding the class of networked devices to physically situated agents, such as robots, autonomous vehicles, and drones~\cite{wangISAC2024,shaoATheorySemanticCommunication2024,wuComprehensiveOverviewOn5GAndBeyond2021}.
Beyond exchanging digital data, these agents can move, reconfigure themselves, manipulate objects, and deliberately alter their surroundings.

When sensing infrastructure interacts with action-capable agents, these physical degrees of freedom can support not only control and task execution, but also communication.
By deliberately selecting a physical state, an agent can encode the intended information that the sensing infrastructure recovers from its observations.
This convergence gives rise to a new communication paradigm: \emph{sensing-induced embodied communication}~\cite{shaoEmbodiedCommunication2026}.

Embodied communication differs fundamentally from both conventional wireless communication and conventional sensing.
Unlike conventional communication, the agent need not encode information into an actively generated radio-frequency (RF) waveform~\cite{weiOrthogonalTimeFrequencySpace2021,ozdemirChannelEstimation2007,shao2021federated,cao2026stepped}.
Unlike conventional sensing, the objective is not to estimate an arbitrary unknown state as accurately as possible~\cite{liuRISAssistedJointSensing2026,panRISAidedNearFieldLocalization2023,wangCooperativeISAC2026}, but to select a set of physically realizable states that can be reliably distinguished from one another.
The communication alphabet is therefore determined jointly by the physical states that the agent can realize and the sensing resolution available to the base station (BS).

This paradigm is particularly relevant when explicit RF transmission is infeasible, undesirable, or redundant.
For example, an agent subject to radio-silence or emission constraints can convey information without generating an RF waveform of its own.
It can also benefit hardware-constrained agents or enable information to be embedded into physical actions already required for task execution.
Although the underlying principle is agnostic to the sensing modality and may apply to visual, acoustic, LiDAR, or multimodal observations, this paper focuses on its RF realization enabled by cellular ISAC.

Existing information-theoretic work on embodied communication considered an agent-controllable region in the far field of the BS~\cite{shaoEmbodiedCommunication2026}.
In this scenario, each message is mapped to a predetermined position of a controllable scatterer, and the BS decodes the message by identifying the selected position from its sensing echoes.
Under the far-field planar-wave propagation, however, the BS array resolves only the scatterer's arrival direction but not its range.
Line-of-sight displacements are therefore indistinguishable, restricting the embodied-symbol space to a two-dimensional (2D) transverse plane.
Moreover, pairwise distinguishability of two controllable positions depends only on the transverse displacement, reducing codebook design to a translation-invariant 2D packing problem.

This far-field assumption becomes inadequate as cellular systems deploy increasingly large antenna arrays in 6G networks~\cite{wangExtremelyLargeScaleMIMO2024,liu2025sensing}.
The corresponding expansion of the Rayleigh distance places relevant controllable regions in the  near field, where spherical wavefront curvature provides both angular and range information~\cite{zhangNearFieldBoundaryDistance2026,leiUnifiedDistributedAlgorithm2026,Cui1,zhangFastNearFieldBeamTraining2022}.
Line-of-sight displacements then become resolvable, allowing all three spatial coordinates to carry information and expanding the embodied-symbol space from a 2D plane to a 3D volume.
This additional spatial degree of freedom raises the central question: \emph{What is the fundamental limit of embodied communication in the near field?}

We formalize this limit through the finite-snapshot $\epsilon$-capacity of the near-field embodied channel, defined from the largest reliable embodied alphabet supported by a finite number of noisy sensing snapshots under a prescribed decoding-error constraint.
The main difficulty lies in the nonuniform reliability geometry induced by near-field propagation.
Unlike the far field, where distinguishability depends only on transverse displacement, near-field distinguishability depends on both the absolute state location and the displacement direction.
The resulting geometry is 3D, position dependent, and anisotropic, precluding a uniform separation rule throughout the controllable region.
Characterizing this geometry and translating it into tractable codebook designs and capacity bounds constitute the central technical challenge of this paper.

Our main contributions are summarized as follows:
\begin{itemize}[leftmargin=0.45cm]
    \item We establish a near-field embodied communication model in which an agent encodes its message into the position of a controllable scatterer within a bounded 3D region, while the BS decodes the message from multi-snapshot monostatic sensing echoes. We then formulate the finite-snapshot $\epsilon$-capacity based on the largest embodied alphabet satisfying a prescribed maximum decoding-error probability.
    \item We quantify the pairwise distinguishability of physical states using the Bhattacharyya distance between their induced sensing distributions and characterize the corresponding 3D confusability regions. Unlike the translation-invariant far-field geometry, these regions vary with the anchor position and are bounded along all three spatial dimensions.
    \item We derive a local anisotropic ellipsoidal representation of each confusability region. Its principal directions and half-axis lengths quantify the sensing resolution along different spatial directions and determine the required separation between embodied symbols. We further show that, as the controllable region moves into the far field, the 3D ellipsoid degenerates into a 2D transverse ellipse, thereby unifying the near-field and far-field reliability geometries.
    \item Using the ellipsoidal reliability geometry, we transform embodied-symbol selection into a 3D packing problem and construct a face-centered cubic (FCC) codebook in the transformed space, yielding an achievable rate and a lower bound on the finite-snapshot $\epsilon$-capacity. We also derive a complementary geometric upper bound. Numerical results validate the proposed geometry and demonstrate the gain of near-field volumetric packing over far-field planar packing.
\end{itemize}


\section{System Model and Problem Formulation}
\label{sec:system}

\subsection{System Model}
\begin{figure}[t]
    \centering
\includegraphics[width=0.9\linewidth]{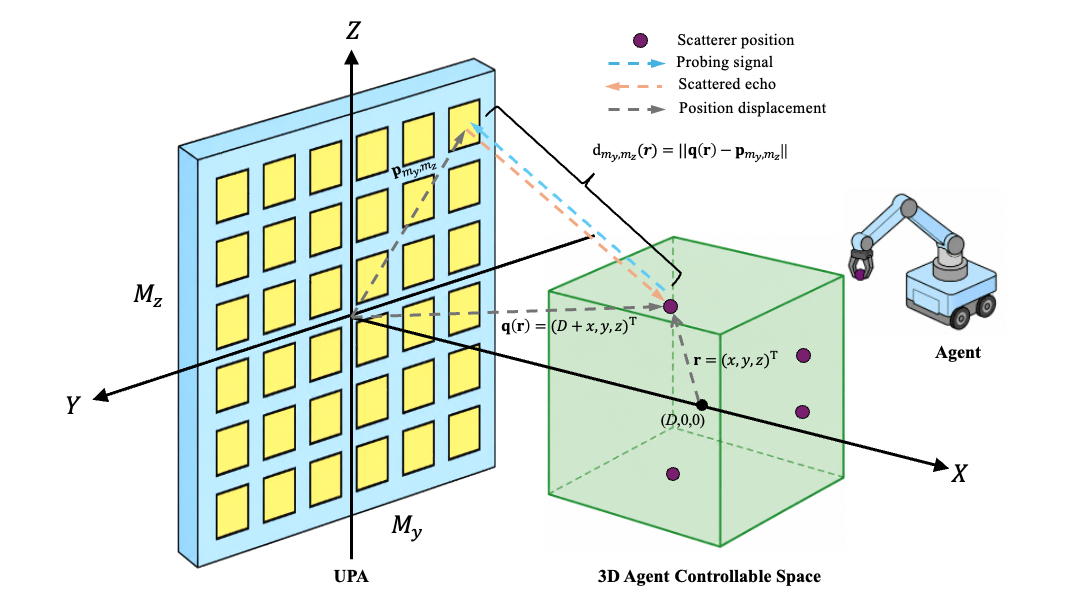}
    \caption{System model of the near-field embodied communication. The agent encodes its message into the physical position of a controllable scatterer within a 3D controllable region. The BS sends probing signals and decodes the intended message by sensing the corresponding echoes from the selected physical state.}
    \label{fig:system_model}
\end{figure}  

We consider a cellular ISAC system comprising a BS and an embodied agent, as shown in Fig.~\ref{fig:system_model}. The BS is equipped with a uniform planar array (UPA), and decodes information from RF sensing observations of a physical state of the agent. The UPA lies in the \(yz\)-plane and is centered at the origin, with its broadside direction aligned with the positive \(x\)-axis. It contains \(M_y\times M_z\) isotropic antenna elements with inter-element spacing \(d=\lambda/2\), where \(\lambda\) is the carrier wavelength. The \((m_y,m_z)\)-th element is located at
\begin{align}
\mathbf p_{m_y,m_z}
\!\!=\!\!
\left[
0,\!
\left(m_y\!-\!\frac{M_y\!-\! 1}{2}\right)d,\! 
\left(m_z\! -\! \frac{M_z\! -\! 1}{2}\right)d
\right]^{\mathrm T},
\label{eq:upa_element_position}
\end{align}
where \(m_y=0,\ldots,M_y-1\) and \(m_z=0,\ldots,M_z-1\).

Unlike conventional uplink communication, the agent does not activate a dedicated RF transmitter~\cite{Zou2024,xiongFundamentalTradeoffIntegrated2023,gaoMultiViewImaging2026}. Instead, it communicates by placing a controllable point scatterer at one of several admissible positions inside a cubical 3D region centered at \((D,0,0)\), where \(D\) is the distance from the BS array center to the center of the controllable region. The admissible displacement set relative to this controllable region center is defined as
\begin{align*}
\mathcal A
=
\left[-\frac{a_x}{2},\frac{a_x}{2}\right]
\times
\left[-\frac{a_y}{2},\frac{a_y}{2}\right]
\times
\left[-\frac{a_z}{2},\frac{a_z}{2}\right]
\subset \R^3,
\end{align*}
where \(a_x\), \(a_y\), and \(a_z\) are the extents along the \(x\), \(y\), and \(z\) axes, respectively. Let \(\mathbf r=(x,y,z)^{\mathrm T}\in\mathcal A\) denote the displacement of the scatterer from the controllable region center, and then the absolute scatterer location is \(\mathbf q(\mathbf r)=[D+x,\,y,\,z]^{\mathrm T}\). To convey information, the agent encodes its message into the physical location of the controllable scatterer. Specifically, a finite set of positions \(\mathcal C=\{\mathbf r_1,\mathbf r_2,\ldots,\mathbf r_J\}\subseteq\mathcal A\) is agreed upon as a \(J\)-ary embodied codebook.
To send the message \(W=j\in\{1,2,\ldots,J\}\) during one embodied channel use, the agent places the scatterer at the corresponding position \(\mathbf r_j\). Information is thus carried by where the scatterer is rather than by any transmitted waveform, realizing the deterministic position-encoding mapping \(W=j\mapsto\mathbf r_j\in\mathcal C\).

The entire controllable region is assumed to lie in the near field of the UPA. Let \(d_{m_y,m_z}(\mathbf r)=\|\mathbf q(\mathbf r)-\mathbf p_{m_y,m_z}\|_2\) denote the propagation distance from the scatterer at \(\mathbf q(\mathbf r)\) to the \((m_y,m_z)\)-th element of the array, and let \(d_0(\mathbf r)=\|\mathbf q(\mathbf r)\|_2\) denote the distance to the array center, and then the $(m_y,m_z)$-th entry of the one-way UPA response is given by
\begin{align}
[\mathbf A(\mathbf r)]_{m_y,m_z}
=
\frac{1}{\sqrt M}
\exp\!\left(
-j\frac{2\pi}{\lambda}
\big(d_{m_y,m_z}(\mathbf r)-d_0(\mathbf r)\big)
\right),
\label{eq:near_field_array_response}
\end{align}
where $M=M_yM_z$ is the total number of the antenna elements, and \(\mathbf A(\mathbf r)\in\mathbb{C}^{M_y\times M_z}\) collects the one-way responses of all UPA elements. Since the system operates in a monostatic sensing setting with colocated transmit and receive arrays, the channel matrix is given by\cite{huaNearField3DLocalization2024,liMIMORadarColocatedAntennas2007}
\begin{align*}
\mathbf H(\mathbf r)=\widetilde{\mathbf a}(\mathbf r)\widetilde{\mathbf a}^{\mathrm T}(\mathbf r),
\end{align*}
where \(\widetilde{\mathbf a}(\mathbf r)=\operatorname{vec}(\mathbf A(\mathbf r))\in\C^{M\times 1}\) is the one-way steering vector, and the corresponding monostatic steering vector is given by 
\begin{align}
\mathbf a(\mathbf r)
=
\widetilde{\mathbf a}(\mathbf r)\otimes\widetilde{\mathbf a}(\mathbf r)
\in\C^{M_{\mathrm v}\times 1},
\label{eq:virtual_monostatic_steering_vector}
\end{align}
where \(M_{\mathrm v}=M^2\).

Each sensing observation comprises \(L\) snapshots. At the \(\ell\)-th snapshot, the BS receives
\begin{align*}
\mathbf y_\ell
=
h_\ell \mathbf a(\mathbf r)+\mathbf n_\ell,
\qquad \ell=1,\ldots,L,
\end{align*}
where \(h_\ell\sim\mathcal{CN}(0,\rho^2)\) denotes the effective random scattering coefficient, and \(\mathbf n_\ell\sim\mathcal{CN}(\mathbf 0,\sigma^2\mathbf I_{M_{\mathrm v}})\) denotes the receiver noise. The coefficients \(\{h_\ell\}_{\ell=1}^{L}\) and noise vectors \(\{\mathbf n_\ell\}_{\ell=1}^{L}\) are assumed mutually independent across snapshots. Therefore, for the selected physical state \(\mathbf r\), each sensing snapshot follows
\begin{align*}
\mathbf y_\ell \mid \mathbf r
\sim
\mathcal{CN}\big(\mathbf 0,\mathbf R(\mathbf r)\big),
\end{align*}
where the covariance matrix  is given by
\begin{align}
\mathbf R(\mathbf r)
=
\rho^2\mathbf a(\mathbf r)\mathbf a^{\mathrm H}(\mathbf r)
+\sigma^2\mathbf I_{M_{\mathrm v}}
=
\sigma^2\!\left(
\mathbf I_{M_{\mathrm v}}+\gamma_0\mathbf a(\mathbf r)\mathbf a^{\mathrm H}(\mathbf r)
\right),
\label{eq:state_covariance}
\end{align}
\textcolor{black}{and \(\gamma_0\triangleq\rho^2/\sigma^2\) is the sensing SNR, defined as the ratio of the average echo power \(\rho^2\) to the per-element noise power \(\sigma^2\).} It can be seen that the selected physical state information is encoded into the covariance matrix through the monostatic steering vector \(\mathbf a(\mathbf r)\).

\subsection{Problem Formulation}

Stack the \(L\)  sensing snapshots as \(\mathbf Y=[\mathbf y_1,\mathbf y_2,\ldots,\mathbf y_L]\in\C^{M_{\mathrm v}\times L}\). Since the observations \(\{\mathbf y_\ell\mid\mathbf r\}\) are mutually independent across the sensing snapshots \(\ell=1,\ldots,L\), and each follows \(\mathcal{CN}(\mathbf 0,\mathbf R(\mathbf r))\), their joint distribution factorizes into a product of \(L\) identical terms as
\begin{align*}
p(\mathbf Y\mid W=j)
=
\prod_{\ell=1}^{L}
\frac{1}{\pi^{M_{\mathrm v}}\det(\mathbf R_j)}
\exp\!\left(
-\mathbf y_\ell^{\mathrm H}\mathbf R_j^{-1}\mathbf y_\ell
\right),
\end{align*}
where \(\mathbf R_j\triangleq \mathbf R(\mathbf r_j)\). Thus, each channel use follows the chain
 \begin{align*}
 W\rightarrow\mathbf r_W\rightarrow\mathbf a(\mathbf r_W)\rightarrow\mathbf R(\mathbf r_W)\rightarrow\mathbf Y.
\end{align*}
Unlike conventional communication channels whose input symbols are waveforms generated in signal space, the input symbols of the embodied channel are physical states carved out of the finite region \(\mathcal A\), and the admissible alphabet is jointly constrained by what the agent can realize and what the BS can distinguish.

For the given embodied codebook \(\mathcal C\), the BS applies a decoder \(g:\C^{M_{\mathrm v}\times L}\rightarrow\{1,2,\ldots,J\}\) to produce the estimate \(\widehat W=g(\mathbf Y)\). The maximum error probability of the codebook under the optimal decoder is given by
\begin{align*}
P_{\max}^{\star}(\mathcal C)
=
\inf_g\ \max_{j\in\{1,\ldots,J\}}
\Pr\big(g(\mathbf Y)\neq j\mid W=j\big),
\end{align*}
where the infimum is taken over all measurable decoders. For a fixed number of sensing snapshots \(L\) and a target error probability \(\epsilon\in(0,1)\), define the largest \(\epsilon\)-reliable embodied alphabet size as
\begin{subequations}
\label{eq:max_reliable_alphabet_size}
\begin{align}
J_\epsilon^{\star}(L)
&=
\underset{\mathcal C\subseteq\mathcal A}{\operatorname{max}}
\quad
|\mathcal C|,
\label{eq:max_reliable_alphabet_size_objective}\\
\text{s.t.}
&\quad
P_{\max}^{\star}(\mathcal C)\le \epsilon.
\label{eq:max_reliable_alphabet_size_constraint}
\end{align}
\end{subequations}
The corresponding finite-snapshot \(\epsilon\)-capacity of the near-field embodied channel is
\begin{align}
C_\epsilon(L)
=
\frac{1}{LT_p}\log_2 J_\epsilon^{\star}(L),
\label{eq:epsilon_capacity}
\end{align}
where \(T_p\) denotes the duration of one probing snapshot. The central question of this paper is therefore how many three-dimensional physical states can be carved out of \(\mathcal A\) and reliably distinguished from \(L\) near-field sensing snapshots.

\section{Near-Field Reliability Metric}
\label{sec:geometry}

Before designing embodied codebooks, we first need to quantify how distinguishable two physical states are from the sensing observations. This section establishes the foundational reliability geometry: a pairwise sensing distance and the anchored 3D confusability 
 region it induces for a given error-probability.

\subsection{Pairwise Bhattacharyya Distance}

We begin with the binary discrimination problem between two candidate positions \(\mathbf r_i,\mathbf r_j\in\mathcal A\). The objective is not to estimate an arbitrary unknown state~\cite{qiaoSensingUsersActivityChannel2024,liu2025sensing,wangCooperativeISAC2026}, but to quantify how reliably two intentionally selected embodied symbols can be distinguished from sensing observations.

For one sensing snapshot, the BS distinguishes
\begin{align}
\mathcal H_i:\mathbf y_\ell\sim\mathcal{CN}(\mathbf 0,\mathbf R_i),
\qquad
\mathcal H_j:\mathbf y_\ell\sim\mathcal{CN}(\mathbf 0,\mathbf R_j).
\label{eq:binary_hypotheses}
\end{align}
To quantify how reliably the two hypotheses in \eqref{eq:binary_hypotheses} can be separated, we measure the statistical overlap between their sensing distributions through the Bhattacharyya coefficient~\cite{kailathDivergenceBhattacharyya2067}~\cite{fukunaga2013introduction}, which is defined as
\begin{align}
\operatorname{BC}(\mathbf r_i,\mathbf r_j)
\triangleq
\int_{\C^{M_{\mathrm v}}}
\sqrt{p(\mathbf y\mid\mathbf r_i)\,p(\mathbf y\mid\mathbf r_j)}
\,d\mathbf y
\in(0,1].
\label{eq:bhattacharyya_coefficient}
\end{align}
It equals one if and only if the two distributions coincide and decreases as they become more separable. Based on the above observation, the following lemma establishes a bound on the pairwise error probability over \(L\) independent sensing snapshots.

\begin{lemma}
\label{lem:bhattacharyya_bound}
Under ML detection between the two hypotheses in \eqref{eq:binary_hypotheses}, the pairwise error probability $P_{i\rightarrow j}$ over \(L\) independent sensing snapshots satisfies
\begin{align}
P_{i\rightarrow j}
\le
\operatorname{BC}(\mathbf r_i,\mathbf r_j)^{L}.
\label{eq:bhattacharyya_bound}
\end{align}
\end{lemma}

\begin{proof}
Please see Appendix~\ref{app:bhattacharyya_bound}.
\end{proof}

The bound \eqref{eq:bhattacharyya_bound} demonstrates that each snapshot shrinks the error bound by the Bhattacharyya coefficient. Accordingly, the Bhattacharyya distance is defined as~\cite{kailathDivergenceBhattacharyya2067}
\begin{align}
B(\mathbf r_i,\mathbf r_j)
\triangleq
-\log\operatorname{BC}(\mathbf r_i,\mathbf r_j)\ge 0,
\label{eq:bhattacharyya_distance_def}
\end{align}
so that \eqref{eq:bhattacharyya_bound} can be rewritten as
\begin{align}
P_{i\rightarrow j}
\le
\exp\!\big[-L B(\mathbf r_i,\mathbf r_j)\big].
\label{eq:pairwise_error_bound}
\end{align}
The distance \(B(\mathbf r_i,\mathbf r_j)\) equals zero if and only if the two states induce identical sensing distributions, and grows as the states become easier to distinguish. In this sense, it compares two physical states through the sensing distributions they induce, and serves as the fundamental reliability exponent governing pairwise discrimination. When the two states follow the Gaussian distribution in \eqref{eq:binary_hypotheses}, the following lemma shows that \(B(\mathbf r_i,\mathbf r_j)\) can be derived in closed form.

\begin{lemma}
\label{lem:gaussian_bhattacharyya} When two received echoes are complex Gaussian distributions as in~\eqref{eq:binary_hypotheses}, the corresponding Bhattacharyya distance is given by
\begin{align}
B(\mathbf r_i,\mathbf r_j)
=
\log
\frac{
\det\!\left((\mathbf R_i+\mathbf R_j)/2\right)
}{
\sqrt{\det(\mathbf R_i)\det(\mathbf R_j)}
}.
\label{eq:bhattacharyya_general}
\end{align}

\end{lemma}

\begin{proof}
Please see Appendix~\ref{app:gaussian_bhattacharyya}.
\end{proof}

Lemma~\ref{lem:gaussian_bhattacharyya} holds for arbitrary Gaussian covariance matrices. For the considered near-field sensing model with covariance matrix~\eqref{eq:state_covariance}, the Bhattacharyya distance can be further derived in a closed form, as stated in the following theorem.

\begin{theorem}
\label{thm:pairwise_reliability}
For any two positions \(\mathbf r_i,\mathbf r_j\in\mathcal A\), let \(\eta(\mathbf r_i,\mathbf r_j)= |\mathbf a^{\mathrm H}(\mathbf r_i)\mathbf a(\mathbf r_j)|^2\) and \(\kappa= \gamma_0^2/[4(1+\gamma_0)]\). The single-snapshot Bhattacharyya distance is
\begin{align}
B(\mathbf r_i,\mathbf r_j)
=
\log\!\left(
1+\kappa\big[1-\eta(\mathbf r_i,\mathbf r_j)\big]
\right),
\label{eq:pairwise_bhattacharyya}
\end{align}
and accordingly, the pairwise error bound \eqref{eq:pairwise_error_bound} over \(L\) snapshots specializes to
\begin{align}
P_{i\rightarrow j}
\le
\Big(
1+\kappa\big[1-\eta(\mathbf r_i,\mathbf r_j)\big]
\Big)^{-L}.
\label{eq:pairwise_error_bound_closed}
\end{align}
\end{theorem}

\begin{proof}
Please see Appendix~\ref{app:pairwise_reliability}.
\end{proof}

Theorem~\ref{thm:pairwise_reliability} shows that the Bhattacharyya distance depends on the two physical states only through the correlation of their steering vectors \(\eta(\mathbf r_i,\mathbf r_j)\). Specifically, the logic chain
\begin{align*}
\eta(\mathbf r_i,\mathbf r_j)\downarrow
&\Rightarrow
B(\mathbf r_i,\mathbf r_j)\uparrow
&\Rightarrow
\exp\!\big[-LB(\mathbf r_i,\mathbf r_j)\big]\downarrow
\Rightarrow
P_{i\to j}\downarrow
\end{align*}
holds, so two positions are reliably distinguished if and only if their monostatic steering vectors are weakly correlated.

\begin{remark}[Near-field breaks displacement-field translation invariance of the far-field]
\label{rem:near_field_two_point}
In the far field, the wavefront  is  planar, so the steering vector depends on the scatterer position only through its arrival direction, whose azimuth and elevation angles are approximately denoted by \(y/D\) and \(z/D\), respectively. This has two implications: First, a displacement along the \(x\)-axis does not change the arrival direction, and hence produces no change in the steering vector that the array could resolve~\cite{liuNearFieldCommunicationsTutorial2023,bjornsonPowerScalingLaws2020}. Therefore, the usable codeword space degenerates from the 3D region \(\mathcal A\) to its 2D transverse cross-section. Second, within this 2D space, as proved in \cite[Eq. (29)]{shaoEmbodiedCommunication2026} the steering-vector correlation \(\eta\) of two positions depends only on their transverse displacement \((\Delta_y,\Delta_z)=(y_i-y_j,\,z_i-z_j)\), so the pairwise reliability exponent collapses to a translation-invariant displacement field \(B_{\mathrm{FF}}(\Delta_y,\Delta_z)\), where all codeword pairs with the same transverse displacement are equally distinguishable, regardless of where they are located.

However, in the near field, both properties do not hold anymore. The wavefront is spherical, and the array elements in \eqref{eq:near_field_array_response} vary with  \(\mathbf r=(x,y,z)^{\mathrm T}\) of the scatterer. Consequently, a displacement along the \(x\)-axis now  becomes distinguishable, so all three coordinates of \(\mathcal A\) can carry information. Moreover, the curvature seen by the array also changes with the absolute position at which a codeword pair is placed, so two pairs with the same 3D displacement \(\mathbf r_i-\mathbf r_j\) can have different reliability exponents depending on where they are located in \(\mathcal A\). The relevant reliability object is therefore not a displacement field but the two-point function \(B(\mathbf r_i,\mathbf r_j)\) over the 3D region.
\end{remark}

\subsection{3D Confusability Region}

The pairwise distance \(B(\mathbf r_i,\mathbf r_j)\) is a two-point measure of distinguishability. To convert it into a spatial resolution representation centered at a single point, fix a reference position \(\mathbf r\in\mathcal A\) as the anchor and describe its neighborhood by its displacement \(\boldsymbol\delta\) from it, which can be denoted as \(\mathbf r+\boldsymbol\delta\). This yields the anchored distance \(B_{\mathbf r}(\boldsymbol\delta)\triangleq B(\mathbf r,\mathbf r+\boldsymbol\delta)\)\footnote{Note that, unlike the far-field displacement field \(B_{\mathrm{FF}}\) in Remark~\ref{rem:near_field_two_point}, the anchored distance still depends on the absolute anchor position through the subscript \(\mathbf r\), and hence the same displacement \(\boldsymbol\delta\) generally yields different distances at different anchors.}.
For a target pairwise error probability \(\epsilon_p\in(0,1)\) after \(L\) snapshots, letting the error bound in \eqref{eq:pairwise_error_bound} be less than \(\epsilon_p\) yields the Bhattacharyya distance threshold
\begin{align}
B_{\mathrm{req}}(\epsilon_p,L)
=
\frac{1}{L}\log\frac{1}{\epsilon_p}.
\label{eq:required_exponent}
\end{align}
Any offset \(\boldsymbol\delta\) with \(B_{\mathbf r}(\boldsymbol\delta)<B_{\mathrm{req}}(\epsilon_p,L)\) fails to meet this threshold and is therefore confusable with the anchor. The near-field confusability region can be written as
\begin{align}
\mathcal S_{\mathrm{conf}}(\mathbf r;\epsilon_p,L)
=
\left\{
\boldsymbol\delta\in\R^3:
\begin{array}{l}
\mathbf r+\boldsymbol\delta\in\mathcal A,\\
B_{\mathbf r}(\boldsymbol\delta)<B_{\mathrm{req}}(\epsilon_p,L)
\end{array}
\right\}.
\label{eq:anchored_confusability_region}
\end{align}
Any displacement outside \(\mathcal S_{\mathrm{conf}}(\mathbf r;\epsilon_p,L)\) is pairwise resolvable at anchor \(\mathbf r\).

\begin{remark}[Position dependent resolution and volumetric alphabets]
\label{rem:anchored_cell}
Since the anchored exponent \(B_{\mathbf r}(\boldsymbol\delta)\) depends on the absolute anchor position, the same displacement \(\boldsymbol\delta\) may be resolvable at one anchor and confusable at another. Near-field resolution is therefore described by a family of position dependent cells \(\mathcal S_{\mathrm{conf}}(\mathbf r;\epsilon_p,L)\). Moreover, since spherical wavefronts make all three coordinates distinguishable, each cell is bounded in all three spatial directions, which enables volumetric embodied alphabets.
\end{remark}

\section{3D Reliability Geometry}
\label{sec:local_metric}

The confusability cell \(\mathcal S_{\mathrm{conf}}(\mathbf r;\epsilon_p,L)\) defined in Section~\ref{sec:geometry} captures the exact near-field resolution geometry, but its dependence on the full two-point field \(B(\mathbf r,\mathbf r+\boldsymbol\delta)\), which makes it difficult to be employed as a codebook design rule. This section derives a tractable local approximation, yielding a position-dependent 3D ellipsoidal geometry that admits a closed-form solution. Moreover, as the agent controllable space moves far away from the array, this ellipsoid degenerates into the far-field  2D ellipse, confirming that the proposed geometry generalizes the far-field scenario.

\subsection{Near-Field 3D Ellipsoidal Geometry}

Once the anchor \(\mathbf r\) is fixed, a tractable local approximation is obtained by expanding \(B_{\mathbf r}(\boldsymbol\delta)\) to second order, as demonstrated in the following proposition.

\begin{proposition}
\label{prop:local_metric}
For sufficiently small \(\boldsymbol\delta\), the Bhattacharyya distance~\eqref{eq:pairwise_bhattacharyya} can be approximated as
\begin{align*}
B_{\mathbf r}(\boldsymbol\delta)
\approx
\boldsymbol\delta^{\mathrm T}
\mathbf G(\mathbf r)
\boldsymbol\delta,
\end{align*}
where 
\begin{align*}
\mathbf G(\mathbf r)
=
\kappa\,
\mathbf J_a^{\mathrm H}(\mathbf r)\mathbf\Pi^\perp_{\mathbf a(\mathbf r)}\mathbf J_a(\mathbf r),
\end{align*}
\(\kappa\) is the same as in Theorem~\ref{thm:pairwise_reliability},
\(\mathbf J_a(\mathbf r)=
\big[
\frac{\partial \mathbf a(\mathbf r)}{\partial x},
\frac{\partial \mathbf a(\mathbf r)}{\partial y},
\frac{\partial \mathbf a(\mathbf r)}{\partial z}
\big]
\in\C^{M_{\mathrm v}\times 3}\)
is the response Jacobian, and
\(\mathbf \Pi^\perp_{\mathbf a(\mathbf r)}=\mathbf I_{M_{\mathrm v}}-\mathbf a(\mathbf r)\mathbf a^{\mathrm H}(\mathbf r)\)
is the projector onto the orthogonal complement of \(\mathbf a(\mathbf r)\).
\end{proposition}

\begin{proof}
Since \(\eta(\mathbf r,\mathbf r+\boldsymbol\delta)=|\mathbf a^{\mathrm H}(\mathbf r)\mathbf a(\mathbf r+\boldsymbol\delta)|^2\), we have
\begin{align}
1-\eta(\mathbf r,\mathbf r+\boldsymbol\delta)
&\overset{(a)}{=}
\|\mathbf a(\mathbf r+\boldsymbol\delta)\|_2^2
-
|\mathbf a^{\mathrm H}(\mathbf r)\mathbf a(\mathbf r+\boldsymbol\delta)|^2 \notag\\
&=
\mathbf a^{\mathrm H}(\mathbf r+\boldsymbol\delta)\,
\mathbf \Pi^\perp_{\mathbf a(\mathbf r)}\,
\mathbf a(\mathbf r+\boldsymbol\delta) \notag\\
&\overset{(b)}{=}
\|\mathbf \Pi^\perp_{\mathbf a(\mathbf r)}\mathbf a(\mathbf r+\boldsymbol\delta)\|_2^2,
\label{eq:projection_correlation_identity}
\end{align}
where \(\mathbf \Pi^\perp_{\mathbf a(\mathbf r)}=\mathbf I_{M_{\mathrm v}}-\mathbf a(\mathbf r)\mathbf a^{\mathrm H}(\mathbf r)\) is the projector onto the orthogonal complement of \(\mathbf a(\mathbf r)\), \((a)\) is derived from \(\|\mathbf a(\mathbf r+\boldsymbol\delta)\|_2=1\), and \((b)\) follows from the  \((\mathbf \Pi^\perp_{\mathbf a(\mathbf r)})^{\mathrm H}\mathbf \Pi^\perp_{\mathbf a(\mathbf r)}=\mathbf \Pi^\perp_{\mathbf a(\mathbf r)}\). For a small displacement \(\boldsymbol\delta\), the first-order expansion of the response is
\begin{align}
\mathbf a(\mathbf r+\boldsymbol\delta)
\approx
\mathbf a(\mathbf r)
+
\mathbf J_a(\mathbf r)\boldsymbol\delta,
\label{eq:array_response_local_expansion}
\end{align}
where \(\mathbf J_a(\mathbf r)=
\big[
\frac{\partial \mathbf a(\mathbf r)}{\partial x},
\frac{\partial \mathbf a(\mathbf r)}{\partial y},
\frac{\partial \mathbf a(\mathbf r)}{\partial z}
\big]
\in\C^{M_{\mathrm v}\times 3}\)
is the response Jacobian.
Substituting \eqref{eq:array_response_local_expansion}
into \eqref{eq:projection_correlation_identity} and noting
\(\mathbf\Pi^\perp_{\mathbf a(\mathbf r)}\mathbf a(\mathbf r)=\mathbf 0\) yields
\begin{align}
1-\eta(\mathbf r,\mathbf r+\boldsymbol\delta)
&\approx
\|\mathbf\Pi^\perp_{\mathbf a(\mathbf r)}\mathbf J_a(\mathbf r)\boldsymbol\delta\|_2^2
\nonumber\\
&\overset{(c)}{=}
\boldsymbol\delta^{\mathrm T}
\mathbf J_a^{\mathrm H}(\mathbf r)\mathbf\Pi^\perp_{\mathbf a(\mathbf r)}\mathbf J_a(\mathbf r)
\boldsymbol\delta,
\label{eq:local_metric_quadratic_form}
\end{align}
where \((c)\) uses
\((\mathbf\Pi^\perp_{\mathbf a(\mathbf r)})^{\mathrm H}\mathbf\Pi^\perp_{\mathbf a(\mathbf r)}=\mathbf\Pi^\perp_{\mathbf a(\mathbf r)}\).
Finally, substituting \eqref{eq:local_metric_quadratic_form} into \eqref{eq:pairwise_bhattacharyya} and applying the first-order approximation \(\log(1+x)\approx x\), we obtain
\begin{align*}
B_{\mathbf r}(\boldsymbol\delta)\approx\boldsymbol\delta^{\mathrm T}\mathbf G(\mathbf r)\boldsymbol\delta,
\end{align*}
with \(\mathbf G(\mathbf r)=\kappa\mathbf J_a^{\mathrm H}(\mathbf r)\mathbf\Pi^\perp_{\mathbf a(\mathbf r)}\mathbf J_a(\mathbf r)\).
\end{proof}

Proposition~\ref{prop:local_metric} shows that the anchored exponent demonstrates a quadratic form through the response Jacobian and the orthogonal projector. The following theorem shows that \(\mathbf G(\mathbf r)\) under the near-field array response can be explicitly derived in closed form.

\begin{theorem} \label{thm:upa_usw_local_metric}
For the near-field array response, the Bhattacharyya distance in Proposition~\ref{prop:local_metric} can be expressed as
\begin{align}
B_{\mathbf r}(\boldsymbol\delta)
\approx
\boldsymbol\delta^{\mathrm T}
\underbrace{
\frac{2\kappa}{M}
\sum_{m=1}^{M}
\mathbf c_m(\mathbf r)
\mathbf c_m^{\mathrm T}(\mathbf r)
}_{\triangleq\,\mathbf G(\mathbf r)}
\boldsymbol\delta,
\label{eq:local_metric_covariance_form}
\end{align}
where \(\mathbf c_m(\mathbf r)\) is the centered phase gradient of the \(m\)-th element, i.e.,
\begin{align*}
\mathbf c_m(\mathbf r)
=
\mathbf u_m(\mathbf r)
-
\frac{1}{M}
\sum_{q=1}^{M}
\mathbf u_q(\mathbf r),
\quad m=1,\ldots,M,
\end{align*}
with \(\mathbf u_m(\mathbf r)=\frac{2\pi}{\lambda}\nabla_{\mathbf r}[d_m(\mathbf r)-d_0(\mathbf r)]\) denoting the per-element phase gradient.
\end{theorem}

\begin{proof}
Based on Proposition~\ref{prop:local_metric}, it suffices to compute
\(\mathbf G(\mathbf r)=\kappa\,\mathbf J_a^{\mathrm H}(\mathbf r)\mathbf\Pi^\perp_{\mathbf a(\mathbf r)}\mathbf J_a(\mathbf r)\). We first derive \(\mathbf J_a(\mathbf r)\) and then compute \(\mathbf G(\mathbf r)\). By the Kronecker structure \eqref{eq:virtual_monostatic_steering_vector}, the \(i_{m,n}\)-th entry (i.e., \(i_{m,n}=(m-1)M+n\)) of \(\mathbf a(\mathbf r)\) is given by
\begin{align}
a_{i_{m,n}}(\mathbf r)
=
[\widetilde{\mathbf a}(\mathbf r)]_m
[\widetilde{\mathbf a}(\mathbf r)]_n
=
\frac{1}{M}\exp[-j(\phi_m(\mathbf r)+\phi_n(\mathbf r))],
\label{eq:app_virtual_scalar_response}
\end{align}
where \(\phi_m(\mathbf r)=\frac{2\pi}{\lambda}[d_m(\mathbf r)-d_0(\mathbf r)]\). Differentiating \eqref{eq:app_virtual_scalar_response} with respect to \(\mathbf r\), the \(i_{m,n}\)-th row of \(\mathbf J_a(\mathbf r)\) is given by
\begin{align}
\left[\mathbf J_a(\mathbf r)\right]_{i_{m,n},:}
=
-j a_{i_{m,n}}(\mathbf r)
\big(\mathbf u_m(\mathbf r)+\mathbf u_n(\mathbf r)\big)^{\mathrm T},
\label{eq:app_jacobian_row}
\end{align}
where \(\mathbf u_m(\mathbf r)=\frac{2\pi}{\lambda}\nabla_{\mathbf r}[d_m(\mathbf r)-d_0(\mathbf r)]\) is the per-element phase gradient.

Then, we compute \(\mathbf G(\mathbf r)\). Since \((\mathbf\Pi^\perp_{\mathbf a(\mathbf r)})^{\mathrm H}\mathbf\Pi^\perp_{\mathbf a(\mathbf r)}=\mathbf\Pi^\perp_{\mathbf a(\mathbf r)}\), we have
\begin{align}
\mathbf G(\mathbf r)
=
\kappa
\big(\mathbf\Pi^\perp_{\mathbf a(\mathbf r)}\mathbf J_a(\mathbf r)\big)^{\mathrm H}
\big(\mathbf\Pi^\perp_{\mathbf a(\mathbf r)}\mathbf J_a(\mathbf r)\big),
\label{eq:app_metric_gram_form}
\end{align}
so it remains to evaluate the term \(\mathbf\Pi^\perp_{\mathbf a(\mathbf r)}\mathbf J_a(\mathbf r)\), which is given by
\begin{align}
\mathbf\Pi^\perp_{\mathbf a(\mathbf r)}\mathbf J_a(\mathbf r)
&=
\mathbf J_a(\mathbf r)
-
\mathbf a(\mathbf r)\,
\mathbf a^{\mathrm H}(\mathbf r)\mathbf J_a(\mathbf r)
\nonumber\\
&\overset{(a)}{=}
-j\mathbf a(\mathbf r)
\Big[
\big(\mathbf u_m(\mathbf r)+\mathbf u_n(\mathbf r)\big)^{\mathrm T}
\nonumber\\
&\qquad
-\frac{1}{M^2}
\sum_{p=1}^{M}\sum_{q=1}^{M}
\big(\mathbf u_p(\mathbf r)+\mathbf u_q(\mathbf r)\big)^{\mathrm T}
\Big]
\nonumber\\
&\overset{(b)}{=}
-j\mathbf a(\mathbf r)
\Big[
\big(\mathbf u_m(\mathbf r)+\mathbf u_n(\mathbf r)\big)^{\mathrm T}
-
\frac{2}{M}
\sum_{q=1}^{M}
\mathbf u_q^{\mathrm T}(\mathbf r)
\Big]
\nonumber\\
&=
-j\mathbf a(\mathbf r)
\big(\mathbf c_m(\mathbf r)+\mathbf c_n(\mathbf r)\big)^{\mathrm T},
\label{eq:projected_jacobian_centered_gradient}
\end{align}
where \((a)\) is derived from substituting \eqref{eq:app_jacobian_row} into both terms and uses \(|a_{i_{p,q}}(\mathbf r)|^2=1/M^2\), \((b)\) is derived from \(M^{-2}\sum_{p,q}\big(\mathbf u_p(\mathbf r)+\mathbf u_q(\mathbf r)\big)=\frac{2}{M}\sum_{q}\mathbf u_q(\mathbf r)\), and the last equality follows from the definition of the centered phase gradient \(\mathbf c_m(\mathbf r)=\mathbf u_m(\mathbf r)-\frac{1}{M}\sum_{q=1}^{M}\mathbf u_q(\mathbf r)\). Finally, substituting \eqref{eq:projected_jacobian_centered_gradient} into \eqref{eq:app_metric_gram_form} yields
\begin{align}
\mathbf G(\mathbf r)
&=
\frac{\kappa}{M^2}
\sum_{m=1}^{M}\sum_{n=1}^{M}
\big(\mathbf c_m(\mathbf r)+\mathbf c_n(\mathbf r)\big)
\big(\mathbf c_m(\mathbf r)+\mathbf c_n(\mathbf r)\big)^{\mathrm T}
\nonumber\\
&\overset{(c)}{=}
\frac{2\kappa}{M}
\sum_{m=1}^{M}
\mathbf c_m(\mathbf r)\mathbf c_m^{\mathrm T}(\mathbf r),
\label{eq:projected_jacobian_covariance_form}
\end{align}
where \((c)\) derives from \(\sum_{m=1}^{M}\mathbf c_m(\mathbf r)=\mathbf 0\). Substituting \eqref{eq:projected_jacobian_covariance_form} into the local expansion \(B_{\mathbf r}(\boldsymbol\delta)\approx\boldsymbol\delta^{\mathrm T}\mathbf G(\mathbf r)\boldsymbol\delta\) of Proposition~\ref{prop:local_metric} gives exactly \eqref{eq:local_metric_covariance_form}, which completes the proof.
\end{proof}

Combining Theorem~\ref{thm:upa_usw_local_metric} with the confusability region definition in \eqref{eq:anchored_confusability_region}, the exact anchored cell is locally approximated by
\begin{align}
\mathcal E_{\mathrm{conf}}(\mathbf r;\epsilon_p,L)
=
\left\{
\boldsymbol\delta\in\R^3:
\boldsymbol\delta^{\mathrm T}\mathbf G(\mathbf r)\boldsymbol\delta
<
B_{\mathrm{req}}(\epsilon_p,L)
\right\}.
\label{eq:local_confusability_ellipsoid}
\end{align}
This quadratic rule is local, but it already shows the shape of reliable
packing near \(\mathbf r\): offsets inside the ellipsoid are
confusable, while offsets outside it meet the target pairwise reliability.
Since \(\mathbf G(\mathbf r)=\frac{2\kappa}{M}\sum_m\mathbf c_m(\mathbf r)\mathbf c_m^{\mathrm T}(\mathbf r)\) is a real symmetric positive semidefinite matrix, it can be decomposed through the eigendecomposition~\cite{moon2000mathematical}
\begin{align*}
\mathbf G(\mathbf r)
=
\mathbf V(\mathbf r)
\operatorname{diag}\{\lambda_1(\mathbf r),\lambda_2(\mathbf r),\lambda_3(\mathbf r)\}
\mathbf V^{\mathrm T}(\mathbf r),
\end{align*}
with \(\lambda_k(\mathbf r)\ge 0\) and \(\mathbf V(\mathbf r)\) is orthogonal. Therefore, the columns \(\mathbf v_k(\mathbf r)\) of \(\mathbf V(\mathbf r)\) are the principal axes of \(\mathcal E_{\mathrm{conf}}\), with corresponding half-lengths given by
\begin{align*}
\ell_k(\mathbf r;\epsilon_p,L)
=
\sqrt{\frac{B_{\mathrm{req}}(\epsilon_p,L)}{\lambda_k(\mathbf r)}},
\qquad k=1,2,3.
\end{align*}
This expression has a direct codebook interpretation. If we move away from
\(\mathbf r\) along a principal direction with a large eigenvalue, the
reliability distance grows rapidly, so a small separation is  enough to
distinguish two codewords. Hence codewords can be placed more densely along
that direction. If the eigenvalue is small, the reliability distance grows slowly; the
ellipsoid stretches in that direction, and codewords must be spaced farther
apart. 

\begin{remark}[Volumetric packing and its degeneration]
\label{rem:volumetric_packing}
In the near field,
\(\mathbf G(\mathbf r)\) is full rank and all three principal half-lengths are
finite, which makes the local codebook genuinely volumetric. When one eigenvalue of \(\mathbf G(\mathbf r)\) is small, the metric turns nearly  rank
deficient, i.e., the corresponding direction is only weakly distinguishable, and the
local geometry is effectively governed by the two-dimensional subspace spanned by
the two dominant eigenvalues. In the limiting case where this eigenvalue
vanishes, \(\mathbf G(\mathbf r)\) becomes rank deficient and the volumetric geometry
degenerates into a planar geometry. This degeneration emerges when the agent
controllable space moves far away from the array, which will be demonstrated in
the next subsection.
\end{remark}

\subsection{Generalization of the Far-Field 2D Ellipse Geometry}

Building on the 3D ellipsoidal geometry established above, this subsection shows that, as the agent controllable space moves far away from the array, displacements along the range direction (i.e., the \(x\)-axis in our setting) become indistinguishable, so that the 3D ellipsoid collapses into the classical 2D far-field ellipse.

\begin{proposition}
\label{prop:far_field_local_metric}
For the fixed array aperture and the side lengths of \(\mathcal A\), when \(D\to\infty\), the local metric in
\eqref{eq:local_metric_covariance_form} is given by
\begin{align}
\mathbf G(\mathbf r)
\approx
\kappa
\begin{bmatrix}
0 & 0 & 0\\
0 & \alpha_y & 0\\
0 & 0 & \alpha_z
\end{bmatrix},
\label{eq:far_field_metric_limit}
\end{align}
where
\begin{align}
\alpha_y
=
\frac{\pi^2(M_y^2-1)}{6D^2},
\qquad
\alpha_z
=
\frac{\pi^2(M_z^2-1)}{6D^2}.
\label{eq:far_field_metric_coefficients}
\end{align}
Consequently, the 3D near-field reliability geometry degenerates into the 2D
far-field geometry, given by
\begin{align}
B_{\mathbf r}(\boldsymbol\delta)
\approx
\kappa
\left(
\alpha_y\delta_y^2+\alpha_z\delta_z^2
\right).
\label{eq:far_field_local_exponent}
\end{align}
\end{proposition}

\begin{proof}
From Theorem~\ref{thm:upa_usw_local_metric}, we have
\begin{align*}
\mathbf G(\mathbf r)
=
\frac{2\kappa}{M}
\sum_{m=1}^{M}
\mathbf c_m(\mathbf r)\mathbf c_m^{\mathrm T}(\mathbf r),
\end{align*}
where
\(\mathbf c_m(\mathbf r)=\mathbf u_m(\mathbf r)-\frac{1}{M}\sum_{q=1}^{M}\mathbf u_q(\mathbf r)\)
and
\(\mathbf u_m(\mathbf r)=\frac{2\pi}{\lambda}\nabla_{\mathbf r}[d_m(\mathbf r)-d_0(\mathbf r)]\).
Hence, to obtain the far-field limit of \(\mathbf G(\mathbf r)\), we first
derive the far-field form of the per-element phase gradient
\(\mathbf u_m(\mathbf r)\), and then obtain
\(\mathbf c_m(\mathbf r)\), and finally evaluate the resulting covariance sum.
 
Let \(p_{m,y}\) and \(p_{m,z}\) denote the \(y\)- and
\(z\)-components of \(\mathbf p_{m_y,m_z}\) in~\eqref{eq:upa_element_position}. For large \(D\), apply the Fresnel approximation
\(\sqrt{(D+x)^2+s}\approx(D+x)+s/[2(D+x)]\), with
\(s=(y-p_{m,y})^2+(z-p_{m,z})^2\) for \(d_m(\mathbf r)\) and
\(s=y^2+z^2\) for \(d_0(\mathbf r)\), we have
\begin{align*}
d_m(\mathbf r)
&\approx
(D+x)
+
\frac{(y-p_{m,y})^2+(z-p_{m,z})^2}{2(D+x)},
\\
d_0(\mathbf r)
&\approx
(D+x)
+
\frac{y^2+z^2}{2(D+x)}.
\end{align*}
Then, we have
\begin{align}
d_m(\mathbf r)-d_0(\mathbf r)
\approx
-\frac{p_{m,y}y+p_{m,z}z}{D+x}
+
\frac{p_{m,y}^2+p_{m,z}^2}{2(D+x)}.
\label{eq:app_far_field_path_expansion}
\end{align}
It can be seen that in the far field, the dependence on x appears only through the denominator \(D+x\). Thus, its derivative is negligible compared with the \(y\) and \(z\) derivatives. Therefore, differentiating \eqref{eq:app_far_field_path_expansion} with respect to \(x\), \(y\), and \(z\) gives
\begin{align*}
\frac{\partial\big[d_m(\mathbf r)-d_0(\mathbf r)\big]}{\partial x}
&=
O\!\big(1/(D+x)^2\big),
\\
\frac{\partial\big[d_m(\mathbf r)-d_0(\mathbf r)\big]}{\partial y}
&=
-\frac{p_{m,y}}{D+x},
\\
\frac{\partial\big[d_m(\mathbf r)-d_0(\mathbf r)\big]}{\partial z}
&=
-\frac{p_{m,z}}{D+x}.
\end{align*}
Therefore, we have
\begin{align*}
\mathbf u_m(\mathbf r)
=
\frac{2\pi}{\lambda}
\nabla_{\mathbf r}
\big[d_m(\mathbf r)-d_0(\mathbf r)\big]
=
-
\frac{2\pi}{\lambda(D+x)}
\begin{bmatrix}
0\\
p_{m,y}\\
p_{m,z}
\end{bmatrix}.
\end{align*}

Moreover, since the UPA in
\eqref{eq:upa_element_position} is centered at the origin, i.e.,
\(\sum_m p_{m,y}=\sum_m p_{m,z}=0\), we have
\(\frac{1}{M}\sum_{q=1}^{M}\mathbf u_q(\mathbf r)=\mathbf 0\),
and thus the centered phase gradient is given by
\begin{align}
\mathbf c_m(\mathbf r)
=
\mathbf u_m(\mathbf r)
-
\frac{1}{M}\sum_{q=1}^{M}\mathbf u_q(\mathbf r)
=
-
\frac{2\pi}{\lambda(D+x)}
\begin{bmatrix}
0\\
p_{m,y}\\
p_{m,z}
\end{bmatrix}.
\label{eq:app_far_field_centered_gradient}
\end{align}
Since \(x\) remains bounded as \(D\to\infty\), we have \(D+x\approx D\).
Substituting \eqref{eq:app_far_field_centered_gradient} into
\eqref{eq:local_metric_covariance_form} yields
\begin{align*}
\mathbf G(\mathbf r)
&\!\!\approx\!\!
\frac{2\kappa}{M}
\left(\frac{2\pi}{\lambda D}\right)^2
\sum_{m=1}^{M}
\begin{bmatrix}
0\\
p_{m,y}\\
p_{m,z}
\end{bmatrix}
\begin{bmatrix}
0 & p_{m,y} & p_{m,z}
\end{bmatrix}
\\
&\!\!=\!\!
\frac{8\pi^2\kappa}{\lambda^2 D^2}
\begin{bmatrix}
0 & 0 & 0\\[2pt]
0 & \frac{1}{M}\sum_{m}p_{m,y}^2 & \frac{1}{M}\sum_{m}p_{m,y}p_{m,z}\\[2pt]
0 & \frac{1}{M}\sum_{m}p_{m,y}p_{m,z} & \frac{1}{M}\sum_{m}p_{m,z}^2
\end{bmatrix}.
\end{align*}

For the diagonal elements, we have
\begin{align*}
\frac{1}{M}\sum_{m=1}^{M}p_{m,y}^2
&=
\frac{d^2}{M_yM_z}
\sum_{m_z=0}^{M_z-1}
\sum_{m_y=0}^{M_y-1}
\left(m_y-\frac{M_y-1}{2}\right)^2
\\
&=
\frac{d^2}{M_y}
\sum_{m_y=0}^{M_y-1}
\left(m_y-\frac{M_y-1}{2}\right)^2
=
\frac{d^2(M_y^2-1)}{12},
\end{align*}
where the last equality follows from the identity \(\sum_{k=0}^{N-1}\big(k-\frac{N-1}{2}\big)^2=\frac{N(N^2-1)}{12}\). Similarly, we have
\(\frac{1}{M}\sum_{m=1}^{M}p_{m,z}^2=\frac{d^2(M_z^2-1)}{12}\).
Moreover, for the off-diagonal elements, we have \(\frac{1}{M}\sum_{m=1}^{M}p_{m,y}p_{m,z}=0\).
Therefore, letting \(d=\lambda/2\), we have
\begin{align*}
\mathbf G(\mathbf r)
\approx
\kappa
\begin{bmatrix}
0 & 0 & 0\\[2pt]
0 & \frac{\pi^2(M_y^2-1)}{6D^2} & 0\\[2pt]
0 & 0 & \frac{\pi^2(M_z^2-1)}{6D^2}
\end{bmatrix},
\end{align*}
which is exactly \eqref{eq:far_field_metric_limit} with \(\alpha_y\) and
\(\alpha_z\) defined in
\eqref{eq:far_field_metric_coefficients}. Finally, substituting
\eqref{eq:far_field_metric_limit} into the local quadratic approximation
\(B_{\mathbf r}(\boldsymbol\delta)\approx
\boldsymbol\delta^{\mathrm T}\mathbf G(\mathbf r)\boldsymbol\delta\) yields
\eqref{eq:far_field_local_exponent}, which completes the
proof.
\end{proof}

Based on Proposition~\ref{prop:far_field_local_metric}, we can see that \(\mathbf G(\mathbf r)\)
is rank deficient along the range direction, so a displacement along \(x\)
does not increase the reliability distance, and the local confusability
ellipsoid in \eqref{eq:local_confusability_ellipsoid} degenerates into the 2D
ellipse
\begin{align*}
\frac{\delta_y^2}{B_{\mathrm{req}}(\epsilon_p,L)/(\kappa\alpha_y)}
+
\frac{\delta_z^2}{B_{\mathrm{req}}(\epsilon_p,L)/(\kappa\alpha_z)}
<
1.
\end{align*}
 Therefore, the 2D ellipse is the transverse
cross-section of the far-field limit of the 3D near-field ellipsoid. Conversely, as the agent controllable space approaches the array, the
wavefront curvature becomes resolvable across the aperture, so the local
metric regains sensitivity to range displacement, and the local reliability
geometry expands from a transverse ellipse to a volumetric ellipsoid.

\section{Near-Field Codebook Design and Capacity Bound}
\label{sec:codebook_design}
The preceding sections characterize the pairwise reliability of physical states. However, a reliable embodied codebook must distinguish every selected state from all its competing states. Therefore, this section lifts the established pairwise reliability criterion to the codebook level design, which constructs a face-centered cubic (FCC) codebook to obtain an achievable lower bound, and derives a geometric packing converse that provides a complementary upper bound.

\subsection{Codebook-Level Reliability}

The pairwise bound \eqref{eq:pairwise_error_bound} controls the confusion between any two states, but reliable communication over codebook \(\mathcal C\) requires the decoder to correctly distinguish the selected state from all \(J-1\) competing states simultaneously. Therefore, based on \eqref{eq:pairwise_error_bound}, a union bound over the \(J-1\) competitors is given by
\begin{align}
P_{\max}^{\star}(\mathcal C)
&\le
\max_{i\in\{1,\ldots,J\}}\sum_{j\ne i}\exp\!\big[-L B(\mathbf r_i,\mathbf r_j)\big]\nonumber\\
&\le
(J-1)\,e^{-L B_{\min}(\mathcal C)},
\label{eq:codebook_union_bound}
\end{align}
where \(B_{\min}(\mathcal C)=\min_{i\ne j}B(\mathbf r_i,\mathbf r_j)\) is the minimum pairwise Bhattacharyya distance of the codebook.  Requiring the right-hand side of \eqref{eq:codebook_union_bound} to equal to the target error probability \(\epsilon\), i.e., \((J-1)\,e^{-L B_{\min}(\mathcal C)}=\epsilon\), we obtain the codebook-level Bhattacharyya threshold
\begin{align}
B_J(\epsilon,L)
=
\frac{1}{L}\log\frac{J-1}{\epsilon}.
\label{eq:codebook_bhattacharyya_threshold}
\end{align}
Compared with the \(B_{\mathrm{req}}(\epsilon_p,L)\) in \eqref{eq:required_exponent}, we have \(B_J(\epsilon,L)=B_{\mathrm{req}}\!\big(\epsilon/(J-1),L\big)\), which means guaranteeing codebook-level reliability is equivalent to requiring each selected state to be pairwise resolvable at the tightened budget \(\epsilon_p=\epsilon/(J-1)\).

Based on the above observation, the codebook \(\mathcal C\) satisfies the \(\epsilon\)-reliability constraint \eqref{eq:max_reliable_alphabet_size_constraint}, i.e., \(P_{\max}^{\star}(\mathcal C)\le\epsilon\), whenever its minimum pairwise Bhattacharyya distance is no smaller than the codebook-level threshold, i.e.,
\begin{align}
B_{\min}(\mathcal C)\ge B_J(\epsilon,L).
\label{eq:codebook_reliability_condition}
\end{align}

\subsection{FCC Codebook Design}

Based on the established ellipsoidal geometry, the reliability condition
\eqref{eq:codebook_reliability_condition} can be evaluated through the
\(\boldsymbol\delta^{\mathrm T}\mathbf G(\mathbf r)\boldsymbol\delta\) in~\eqref{eq:local_metric_covariance_form}. To absorb the local approximation error, we introduce a penalty factor \(\xi\ge0\) with
\begin{equation}
\label{eq:margin_forbidden_ellipsoid}
\boldsymbol\delta^{\mathrm T}\mathbf G(\mathbf r)\boldsymbol\delta
\ge
(1+\xi)B_J(\epsilon,L).
\end{equation}
Since \(\mathbf G(\mathbf r)\succeq\mathbf 0\), we can
define the transformation \(\widetilde{\boldsymbol\delta}=\mathbf G^{1/2}(\mathbf r)\boldsymbol\delta\).
Then, \eqref{eq:margin_forbidden_ellipsoid} is equivalently transformed as
\begin{align}
\|\widetilde{\boldsymbol\delta}\|_2\ge d_J,
\quad\text{s.t.}\quad
d_J = \sqrt{(1+\xi)B_J(\epsilon,L)}.
\label{eq:fcc_minimum_distance}
\end{align}
Hence, the codebook design in the transformed space becomes a packing problem under the minimum-distance requirement~\(d_J\). As the FCC lattice achieves the densest packing in three-dimensional space~\cite{conway2013sphere},  the  corresponding lattice generator is given by
\begin{align}
\mathbf G_{\mathrm{FCC}}(d_J)
=
\frac{d_J}{\sqrt 2}
\begin{bmatrix}
0 & 1 & 1\\
1 & 0 & 1\\
1 & 1 & 0
\end{bmatrix},
\label{eq:fcc_generator}
\end{align}
which spans the FCC lattice
\(\Lambda_{\mathrm{FCC}}(d_J)=\{\mathbf G_{\mathrm{FCC}}(d_J)\mathbf n:
\mathbf n\in\mathbb Z^3\}\). 
Since the shortest nonzero vector of \(\Lambda_{\mathrm{FCC}}(d_J)\) has norm
\(d_J\), every pair of codewords satisfies the required separation condition~\eqref{eq:fcc_minimum_distance}.

To determine the codebook size, it remains to count the lattice points within the transformed controllable region \(\widetilde{\mathcal A}\). Since the local transformation \(\mathbf G^{1/2}(\mathbf r)\) scales the volume element at \(\mathbf r\) by \(\sqrt{\det\mathbf G(\mathbf r)}\), the effective volume of \(\widetilde{\mathcal A}\) can be computed as
\begin{align}
\mathcal V_G
=
\int_{\mathcal A}\sqrt{\det\mathbf G(\mathbf r)}\,d\mathbf r.
\label{eq:metric_volume}
\end{align}
Meanwhile, based on~\eqref{eq:fcc_generator}, each FCC
codeword occupies one fundamental cell space of the volume
\begin{align}
V_{\mathrm{cell}}(d_J)
=
\left|\det \mathbf G_{\mathrm{FCC}}(d_J)\right|
=
\frac{d_J^3}{\sqrt 2}.
\label{eq:fcc_cell_volume}
\end{align}
Therefore, combining \eqref{eq:metric_volume} and \eqref{eq:fcc_cell_volume} and ignoring boundary effects, the codebook size is given by
\begin{equation}
\label{eq:implicit_fcc_size}
J
\approx
\frac{\mathcal V_G}{V_{\mathrm{cell}}(d_J)}
=
\frac{
\sqrt 2\,\int_{\mathcal A}\sqrt{\det\mathbf G(\mathbf r)}\,d\mathbf r
}{
\big[(1+\xi)B_J(\epsilon,L)\big]^{3/2}
}.
\end{equation}
However, \(B_J(\epsilon,L)\) itself depends on \(J\),
\eqref{eq:implicit_fcc_size} is an implicit equation in \(J\).  The following theorem demonstrates that the 
achievable codebook size admits a closed-form independent of \(J\),  and thereby establishes an achievable lower bound on the \(\epsilon\)-capacity.

\begin{theorem}
\label{thm:fcc_design}
The FCC codebook achieves the alphabet size
\begin{align}
J_{\mathrm{FCC}}(L,\epsilon,\xi)
&\approx
\left\lfloor
\frac{
\left(\frac{2}{3}\right)^{3/2}\Psi(\xi)
}{
\left[
W_0\!\left(
\frac{2}{3}
\left(\frac{\Psi(\xi)}{\epsilon}\right)^{2/3}
\right)
\right]^{3/2}
}
\right\rfloor,
\label{eq:fcc_closed_form_size}
\end{align}
where
\(\Psi(\xi)=
\sqrt 2\,\mathcal V_G L^{3/2}/(1+\xi)^{3/2}\), and \(W_0(\cdot)\) denotes
the principal branch of the Lambert-\(W\) function. Hence, the corresponding
achievable rate is given by
\begin{align}
R_{\mathrm{FCC}}(L,\epsilon,\xi)
=
\frac{1}{LT_p}\log_2 J_{\mathrm{FCC}}(L,\epsilon,\xi),
\label{eq:fcc_capacity_lower_bound}
\end{align}
which is an achievable lower bound on the \(\epsilon\)-capacity $C_\epsilon(L)$.
\end{theorem}

\begin{proof}
For sufficiently large \(J\), the codebook-level threshold
\eqref{eq:codebook_bhattacharyya_threshold} satisfies
\begin{align*}
B_J(\epsilon,L)
=
\frac{1}{L}\log\frac{J-1}{\epsilon}
\approx
\frac{1}{L}\log\frac{J}{\epsilon}.
\end{align*}
Substituting this approximation into
\eqref{eq:implicit_fcc_size} gives
\begin{align}
\label{eq:implicit_fcc_size_approximation}
J
&\approx
\frac{\Psi(\xi)}{[\log(J/\epsilon)]^{3/2}},
\end{align}
where \(\Psi(\xi)=
\sqrt 2\,\mathcal V_G L^{3/2}/(1+\xi)^{3/2}\).
Let \(u=\log(J/\epsilon)\), so that \(J=\epsilon e^u\). Then, the 
equation~\eqref{eq:implicit_fcc_size_approximation} can be rewritten as
\begin{align*}
\epsilon e^u u^{3/2}
=
\Psi(\xi).
\end{align*}
Raising both sides to the power \(2/3\) and rearranging yields
\begin{align*}
\frac{2u}{3}e^{2u/3}
=
\frac{2}{3}
\left(\frac{\Psi(\xi)}{\epsilon}\right)^{2/3}.
\end{align*}
By the definition of the principal Lambert-\(W\) function, we have
\begin{align}
\label{eq:implicit_fcc_size_solution}
u
=
\frac{3}{2}W_0\!\left(
\frac{2}{3}
\left(\frac{\Psi(\xi)}{\epsilon}\right)^{2/3}
\right).
\end{align}
Substituting~\eqref{eq:implicit_fcc_size_solution} into
\eqref{eq:implicit_fcc_size_approximation}, i.e.,
\(J\approx \Psi(\xi)/u^{3/2}\), yields \eqref{eq:fcc_closed_form_size}. Finally, substituting
\eqref{eq:fcc_closed_form_size} into the rate expression
\eqref{eq:epsilon_capacity} yields \eqref{eq:fcc_capacity_lower_bound}, which
completes the proof.
\end{proof}

\begin{remark}[Generalization to far-field hexagonal packing]
\label{rem:far_field_hexagonal_packing}
Proposition~\ref{prop:far_field_local_metric} demonstrates
that, in the far field,
the range direction becomes indistinguishable and the position-dependent 3D
metric reduces to the position-invariant transverse metric 
\begin{align*}
\mathbf G_{\mathrm{FF}}
&=
\kappa\operatorname{diag}\{\alpha_y,\alpha_z\}.
\end{align*}
Accordingly, the usable codeword region reduces to 2D area in the \(yz\)-direction, and the corresponding
 transformation \(\widetilde{\mathbf r}_{\mathrm{FF}}=
\mathbf G_{\mathrm{FF}}^{1/2}\mathbf r_{\mathrm{FF}}\) maps the 2D controllable area to
\begin{align*}
\mathcal V_{\mathrm{FF}}
&=
a_ya_z\sqrt{\det\mathbf G_{\mathrm{FF}}}
=
\kappa a_ya_z\sqrt{\alpha_y\alpha_z}.
\end{align*}
The 3D packing problem therefore reduces to a 2D packing problem, for which the densest packing is the hexagonal lattice~\cite{conway2013sphere,zamir2014lattice}. For the
minimum distance \(d_J\) in \eqref{eq:fcc_minimum_distance}, the corresponding generator and
its fundamental-cell area are given by
\begin{align*}
\mathbf G_{\mathrm{hex}}(d_J)
&=
d_J
\begin{bmatrix}
1 & 1/2\\
0 & \sqrt 3/2
\end{bmatrix},
&
V_{\mathrm{hex}}(d_J)
&=
\frac{\sqrt 3}{2}d_J^2.
\end{align*}
Therefore, the codebook size is given by
\begin{align*}
J_{\mathrm{hex}}
&\approx
\frac{\mathcal V_{\mathrm{FF}}}{V_{\mathrm{hex}}(d_J)}
=
\frac{2\kappa a_ya_z\sqrt{\alpha_y\alpha_z}}
{\sqrt 3(1+\xi)B_J(\epsilon,L)}.
\end{align*}
Similar to the proof of Theorem~\ref{thm:fcc_design}, the corresponding far-field codebook size and
achievable rate can be derived as
\begin{align*}
J_{\mathrm{hex}}(L,\epsilon,\xi)
&\approx
\left\lfloor
\frac{\Psi_{\mathrm{FF}}(\xi)}
{W_0\!\left(\Psi_{\mathrm{FF}}(\xi)/\epsilon\right)}
\right\rfloor,
\\
R_{\mathrm{hex}}(L,\epsilon,\xi)
&=
\frac{1}{LT_p}\log_2 J_{\mathrm{hex}}(L,\epsilon,\xi).
\end{align*}
where \(\Psi_{\mathrm{FF}}(\xi)=
2\kappa a_ya_z\sqrt{\alpha_y\alpha_z}L/[\sqrt 3(1+\xi)]\). 
Since the same controllable region supports 3D packing in the near field but only 2D packing in the far field, the additional range resolution of the near-field implies a achievable capacity gain.
\end{remark}

\subsection{Geometric Converse}

The FCC codebook establishes an achievable lower bound on the \(\epsilon\)-capacity. We now derive an upper bound from a necessary condition on any \(\epsilon\)-reliable codebook, i.e., every pair of codewords must remain \(\epsilon\)-distinguishable even in the absence of all others. Therefore, we first derive this pairwise separation requirement from optimal binary detection, and then convert it into a geometric packing bound on the codebook design.

To this end, consider the binary discrimination between two positions
\(\mathbf r_i,\mathbf r_j\in\mathcal A\) under equal priors, and let
\(P_b^{\star}\) denote the minimum Bayes error probability over all
binary tests. Since no detector can achieve a smaller error probability, \(P_b^{\star}>\epsilon\) implies that the two positions cannot be distinguished at the target reliability. The optimal binary Bayes error is given by~\cite[Eq. (44)]{nielsen2014generalized}
\begin{align*}
P_b^{\star}
&=
\frac{1}{2}\left[
1-\left\|
p(\mathbf Y\mid\mathbf r_i)-p(\mathbf Y\mid\mathbf r_j)
\right\|_{\mathrm{TV}}
\right],
\end{align*}
where
\(\|\cdot\|_{\mathrm{TV}}=\frac{1}{2}\int|\cdot|\,d\mathbf Y\)
denotes the total variation distance. The following lemma lower-bounds \(P_b^{\star}\) in terms of the
Bhattacharyya distance.

\begin{lemma}
\label{lem:binary_bayes_bhattacharyya}
For any two positions \(\mathbf r_i,\mathbf r_j\in\mathcal A\), the minimum
binary Bayes error probability \(P_b^{\star}\) over \(L\) independent
snapshots satisfies
\begin{align}
P_b^{\star}
&\ge
\frac{1}{2}
\left(
1-\sqrt{1-\exp[-2LB(\mathbf r_i,\mathbf r_j)]}
\right).
\label{eq:binary_bayes_bhattacharyya_lower_bound}
\end{align}
\end{lemma}

\begin{proof}
The total variation distance between the two \(L\)-snapshot densities can be
bounded as
\begin{align*}
&\left\|
p(\mathbf Y\mid\mathbf r_i)-p(\mathbf Y\mid\mathbf r_j)
\right\|_{\mathrm{TV}}
\\
&=
\frac{1}{2}\int
\Big|\sqrt{p(\mathbf Y\mid\mathbf r_i)}-\sqrt{p(\mathbf Y\mid\mathbf r_j)}\Big|
\\
&\qquad\qquad\times
\Big(\sqrt{p(\mathbf Y\mid\mathbf r_i)}+\sqrt{p(\mathbf Y\mid\mathbf r_j)}\Big)
\,d\mathbf Y
\\
&\overset{(a)}{\le}
\frac{1}{2}
\left[\int
\Big(\sqrt{p(\mathbf Y\mid\mathbf r_i)}-\sqrt{p(\mathbf Y\mid\mathbf r_j)}\Big)^2
d\mathbf Y\right]^{1/2}
\\
&\qquad\quad\times
\left[\int
\Big(\sqrt{p(\mathbf Y\mid\mathbf r_i)}+\sqrt{p(\mathbf Y\mid\mathbf r_j)}\Big)^2
d\mathbf Y\right]^{1/2}
\\
&=
\left[
1-\left(
\int\sqrt{p(\mathbf Y\mid\mathbf r_i)\,p(\mathbf Y\mid\mathbf r_j)}\,d\mathbf Y
\right)^2
\right]^{1/2},
\\
&\overset{(b)}{=}
\sqrt{1-\exp[-2LB(\mathbf r_i,\mathbf r_j)]}
\end{align*}
where \((a)\) follows from the Cauchy--Schwarz inequality and \((b)\) follows from the definition of the Bhattacharyya distance in \eqref{eq:bhattacharyya_distance_def}.
 Substituting the above upper bound
into \(P_b^{\star}=\frac{1}{2}\big[1-\|p(\mathbf Y\mid\mathbf r_i)-p(\mathbf
Y\mid\mathbf r_j)\|_{\mathrm{TV}}\big]\) yields
\eqref{eq:binary_bayes_bhattacharyya_lower_bound}, which completes the proof.
\end{proof}

For \(0<\epsilon<1/2\), the lower bound in
\eqref{eq:binary_bayes_bhattacharyya_lower_bound} is strictly larger than
\(\epsilon\) whenever
\begin{align}
B(\mathbf r_i,\mathbf r_j)
<
B_{\mathrm{nec}}(\epsilon,L)
=
\frac{1}{2L}
\log\frac{1}{4\epsilon(1-\epsilon)},
\label{eq:necessary_bhattacharyya_threshold}
\end{align}
and hence no detector can distinguish the two positions at the target
reliability. Therefore, every codeword pair in an \(\epsilon\)-reliable
codebook must satisfy
\(B(\mathbf r_i,\mathbf r_j)\ge B_{\mathrm{nec}}(\epsilon,L)\).

To convert this Bhattacharyya separation into a geometric packing
constraint, we introduce a uniform Euclidean exclusion radius over
\(\mathcal A\), defined as the minimum separation at which the necessary
threshold \eqref{eq:necessary_bhattacharyya_threshold} can be satisfied, i.e.,
\begin{align}
d_{\mathrm{nec}}(\epsilon,L)
=
\inf_{\substack{
\mathbf r,\,\mathbf r+\boldsymbol\delta\in\mathcal A:\\
B_{\mathbf r}(\boldsymbol\delta)\ge B_{\mathrm{nec}}(\epsilon,L)
}}
\|\boldsymbol\delta\|_2.
\label{eq:necessary_euclidean_radius}
\end{align}
Hence, two positions closer than \(d_{\mathrm{nec}}(\epsilon,L)\) cannot
coexist in an \(\epsilon\)-reliable codebook. Moreover, applying the local
quadratic approximation \eqref{eq:local_metric_covariance_form} gives
\(B_{\mathbf r}(\boldsymbol\delta)
\approx
\boldsymbol\delta^{\mathrm T}\mathbf G(\mathbf r)\boldsymbol\delta
\le
\lambda_{\max}^{\mathrm{sup}}\|\boldsymbol\delta\|_2^2\), where
\(\lambda_{\max}^{\mathrm{sup}}=
\sup_{\mathbf r\in\mathcal A}\lambda_{\max}(\mathbf G(\mathbf r))\), which
yields an estimate of this exclusion radius.
\begin{align*}
d_{\mathrm{nec}}(\epsilon,L)
\approx
\sqrt{
\frac{B_{\mathrm{nec}}(\epsilon,L)}{\lambda_{\max}^{\mathrm{sup}}}
}.
\end{align*}

Based on the established exclusion radius \(d_{\mathrm{nec}}(\epsilon,L)\),
the converse reduces to a sphere-packing problem, and the following theorem demonstrates that the
capacity upper bound can be derived in closed form.

\begin{theorem}
\label{thm:geometric_packing_converse}
The upper bound on the size of any
\(\epsilon\)-reliable codebook is given by
\begin{align}
\Gamma_{\mathrm{upp}}(\epsilon,L)
=&\,
1
+\frac{3S_1}{2d_{\mathrm{nec}}(\epsilon,L)}
+\frac{6S_2}{\pi d_{\mathrm{nec}}^{2}(\epsilon,L)}
\nonumber\\
&+\frac{6V_{\mathcal A}}{\pi d_{\mathrm{nec}}^{3}(\epsilon,L)},
\label{eq:geometric_packing_bound}
\end{align}
where \(S_1=a_x+a_y+a_z\),
\(S_2=a_xa_y+a_xa_z+a_ya_z\), and \(V_{\mathcal A}=a_xa_ya_z\). Hence, the
corresponding upper bound on the \(\epsilon\)-capacity \(C_\epsilon(L)\) is
given by
\begin{align}
R_{\mathrm{upp}}(L,\epsilon)
=
\frac{1}{LT_p}
\log_2\Gamma_{\mathrm{upp}}(\epsilon,L).
\label{eq:geometric_capacity_upper_bound}
\end{align}
\end{theorem}

\begin{proof}
Let \(d_{\mathrm{nec}}=d_{\mathrm{nec}}(\epsilon,L)\) and
\(r=d_{\mathrm{nec}}/2\). Based on
\eqref{eq:necessary_euclidean_radius}, distinct codewords are separated by at
least \(2r\). Hence, the radius-\(r\) balls centered at the codewords have
disjoint interiors. Since every codeword lies in \(\mathcal A\), all these
balls are contained in the expanded region
\(\mathcal A_{\mathrm{exp}}(r)\triangleq
\{\mathbf u+\mathbf v:\mathbf u\in\mathcal A,\ \|\mathbf v\|_2\le r\}\),
and therefore
\begin{align}
J\frac{4\pi r^3}{3}
\le
\operatorname{Vol}\!\left(\mathcal A_{\mathrm{exp}}(r)\right).
\label{eq:app_sphere_packing_volume}
\end{align}
For the rectangular region \(\mathcal A\), the expanded region consists of
the original box, six face slabs, twelve quarter-cylinder edge pieces, and
eight spherical octants. Their total volume is
\begin{align}
\operatorname{Vol}\!\left(\mathcal A_{\mathrm{exp}}(r)\right)
=
V_{\mathcal A}
+2S_2r
+\pi S_1r^2
+\frac{4\pi r^3}{3},
\label{eq:app_box_parallel_volume}
\end{align}
where \(V_{\mathcal A}=a_xa_ya_z\), \(S_2=a_xa_y+a_xa_z+a_ya_z\), and \(S_1=a_x+a_y+a_z\). Substituting \eqref{eq:app_box_parallel_volume} into
\eqref{eq:app_sphere_packing_volume} and letting
\(r=d_{\mathrm{nec}}/2\) yields
\begin{align*}
J
&\le
1
+\frac{3S_1}{2d_{\mathrm{nec}}}
+\frac{6S_2}{\pi d_{\mathrm{nec}}^2}
+\frac{6V_{\mathcal A}}{\pi d_{\mathrm{nec}}^3}
=
\Gamma_{\mathrm{upp}}(\epsilon,L),
\end{align*}
which drives the~\eqref{eq:geometric_packing_bound}. Substituting
this bound into \eqref{eq:epsilon_capacity} yields
\eqref{eq:geometric_capacity_upper_bound}, which completes the proof.
\end{proof}

\section{Numerical Results}
\label{sec:numerics}
This section numerically evaluates the near-field reliability geometry and the achievable performance of the proposed embodied communication system. Unless otherwise specified, the simulation parameters are set as follows. The carrier frequency is \(f_c=6\)~GHz, corresponding to the wavelength \(\lambda=5\)~cm. The BS is equipped with a \(32\times32\) UPA with half-wavelength spacing. The corresponding physical aperture is \(0.78\)~m per side, and the array diagonal is approximately \(1.10\)~m, which yields a Rayleigh distance of 48~m~\cite{leiUnifiedDistributedAlgorithm2026,Cui1}. The center of the agent-controllable region is placed at \(D=5\)~m from the BS array center. The controllable region is a cube with \(a_x=a_y=a_z=2\)~m. The sensing SNR is \(\gamma_0=20\)~dB, the default number of sensing snapshots is \(L=5\), and the target maximum decoding error probability is \(\epsilon=10^{-3}\). We consider the normalized rate \(R=(1/L)\log_2 J\) in bits per probing duration.

\subsection{Position-dependent Anisotropic Near-field Reliability}

Fig.~\ref{fig:range_displacement_reliability} illustrates the resolvability of range displacements, where the exact Bhattacharyya distance \(B(\mathbf r_u,\mathbf r_u+s\mathbf e_x)\) in \eqref{eq:pairwise_bhattacharyya} is plotted versus the range displacement \(s\in[-1.0,1.0]\)~m at the diagonal anchors \(\mathbf r_u=(0,u,u)^{\mathrm T}\) with \(u\in\{0,0.25,0.5,0.75,1.0\}\)~m, and the pairwise threshold is \(B_{\mathrm{req}}=\log(1/\epsilon_p)/L=1.3816\) with \(\epsilon_p=10^{-3}\). It can be seen that as the anchor moves away from the \(x\)-axis, the Bhattacharyya distance increases and the threshold crossings move closer to the anchor, from \(s\approx-0.53\)~m and \(s\approx+0.66\)~m at \(u=0.25\)~m to \(s\approx\pm0.17\)~m at \(u=1.0\)~m. These results demonstrate that near-field range resolvability is determined by the absolute anchor position.

Fig.~\ref{fig:transverse_displacement_reliability} further demonstrates the resolvability of transverse displacements, where \(B(\mathbf r_v,\mathbf r_v+s\mathbf e_d)\) is plotted versus the displacement \(s\in[-0.3,0.3]\)~m along the \(yz\)-diagonal direction \(\mathbf e_d=(\mathbf e_y+\mathbf e_z)/\sqrt{2}\) at the axial anchors \(\mathbf r_v=(v,0,0)^{\mathrm T}\) with \(v\in\{-1.0,-0.5,0,0.5,1.0\}\)~m. It is observed that the threshold is crossed at \(\lvert s\rvert\approx36\)--\(53\)~mm, which is much smaller than the range crossing distances in Fig.~\ref{fig:range_displacement_reliability} and thus reveals the anisotropy of the near-field reliability geometry. Moreover, the crossing distance decreases monotonically as \(v\) decreases, since a scatterer closer to the array subtends a larger aperture angle and is therefore easier to be resolved transversely.

\begin{figure}[t]
    \centering
    \includegraphics[width=0.9\linewidth]{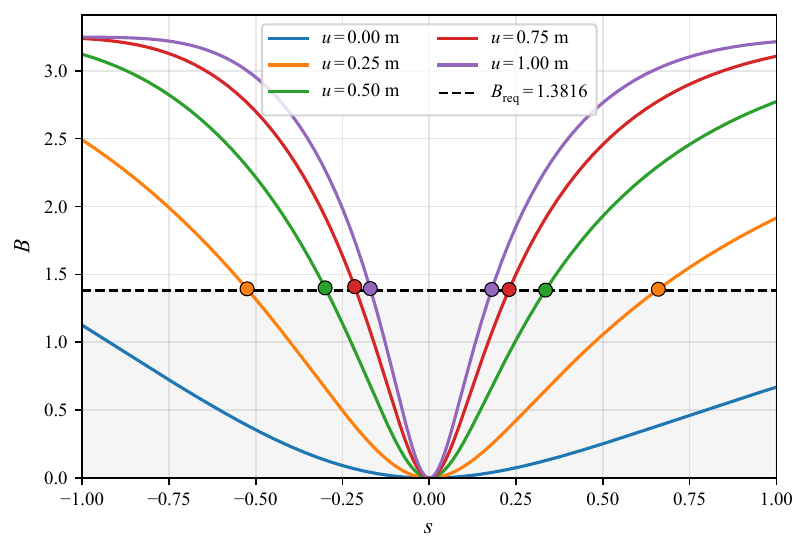}
    \caption{Bhattacharyya distance versus the range displacement at different diagonal anchors.}
\label{fig:range_displacement_reliability}
\end{figure}

\begin{figure}[t]
    \centering
    \includegraphics[width=0.9\linewidth]{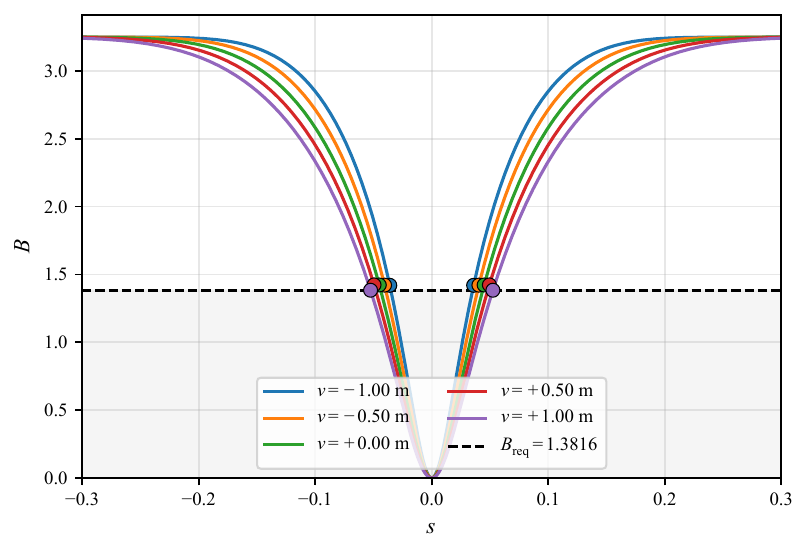}
    \caption{Bhattacharyya distance versus the transverse displacement at different axial anchors.}
\label{fig:transverse_displacement_reliability}
\end{figure}

\subsection{3D Ellipsoidal Reliability Geometry}

Fig.~\ref{fig:local_quadratic_accuracy} validates the accuracy of the ellipsoidal approximation in Proposition~\ref{prop:local_metric}, where the exact Bhattacharyya distance \(B(\mathbf r,\mathbf r+d\mathbf e_k)\) and the quadratic form \(d^2\mathbf e_k^{\mathrm T}\mathbf G(\mathbf r)\mathbf e_k\) are compared along the representative Cartesian directions \(k\in\{x,y\}\) at three anchors \(\mathbf r_1=(-0.8,0.8,0.2)^{\mathrm T}\)~m, \(\mathbf r_2=(0,0.5,0.5)^{\mathrm T}\)~m, and \(\mathbf r_3=(0.8,0.2,0.8)^{\mathrm T}\)~m. The results show that the two curves closely match in both directions at all three anchors. Moreover, the curves vary across the three anchors, which again demonstrates that the near-field reliability metric depends on the absolute anchor position rather than the displacement alone.

\begin{figure}[t!]
    \centering
\includegraphics[width=0.9\linewidth]{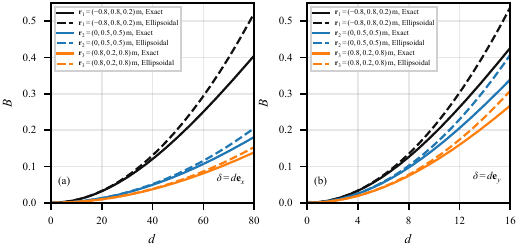}
    \caption{Exact Bhattacharyya distance and its local ellipsoidal approximation at different anchors along the x-direction and the y-direction.}
    \label{fig:local_quadratic_accuracy}
\end{figure}

\begin{figure}
    \centering
    \begin{subfigure}[t]{0.85\linewidth}
        \centering
\includegraphics[width=0.75\linewidth]{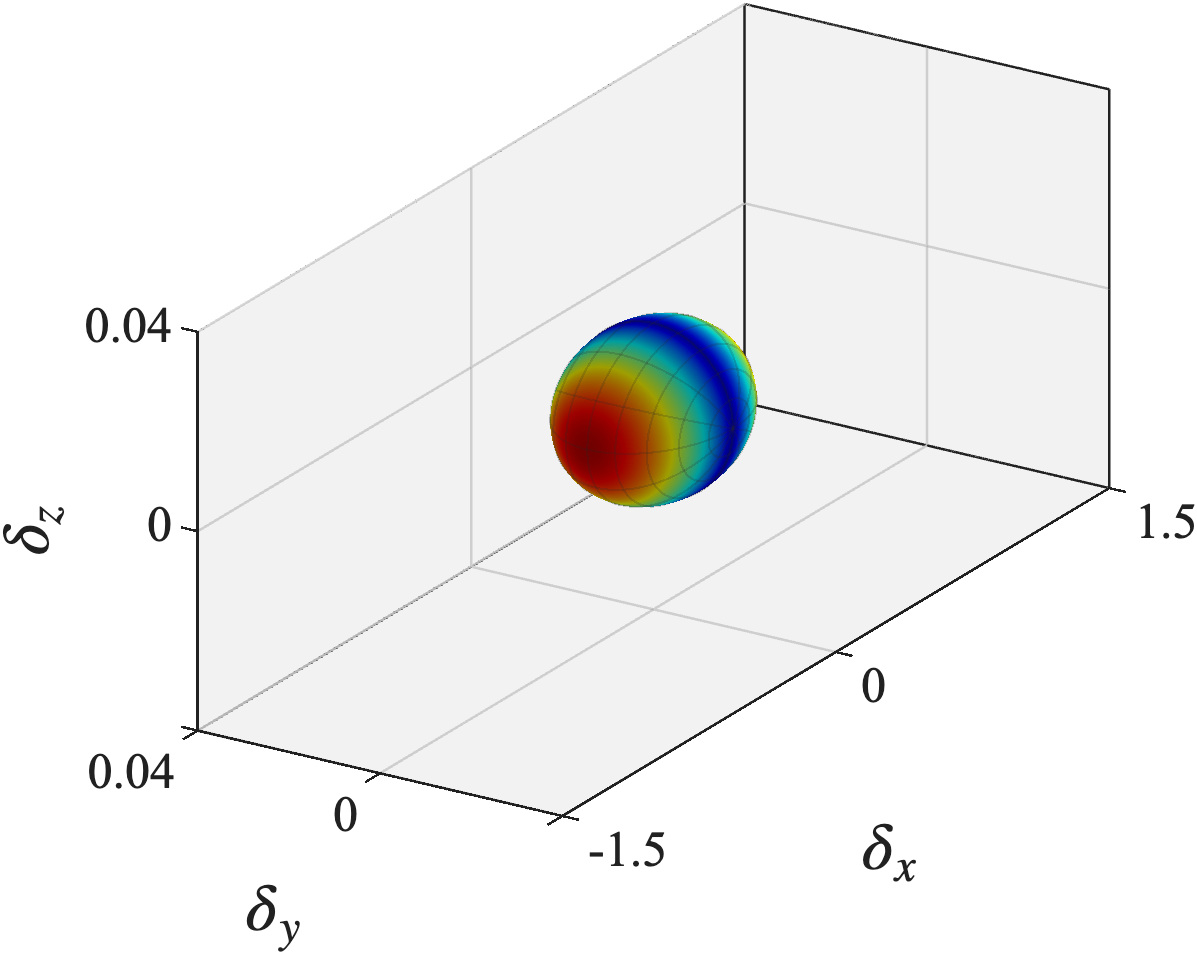}
        \caption{\(D=3\)~m.}
    \end{subfigure}
    \vspace{-0.6em}
    \begin{subfigure}[t]{0.85\linewidth}
        \centering
        \includegraphics[width=0.75\linewidth]{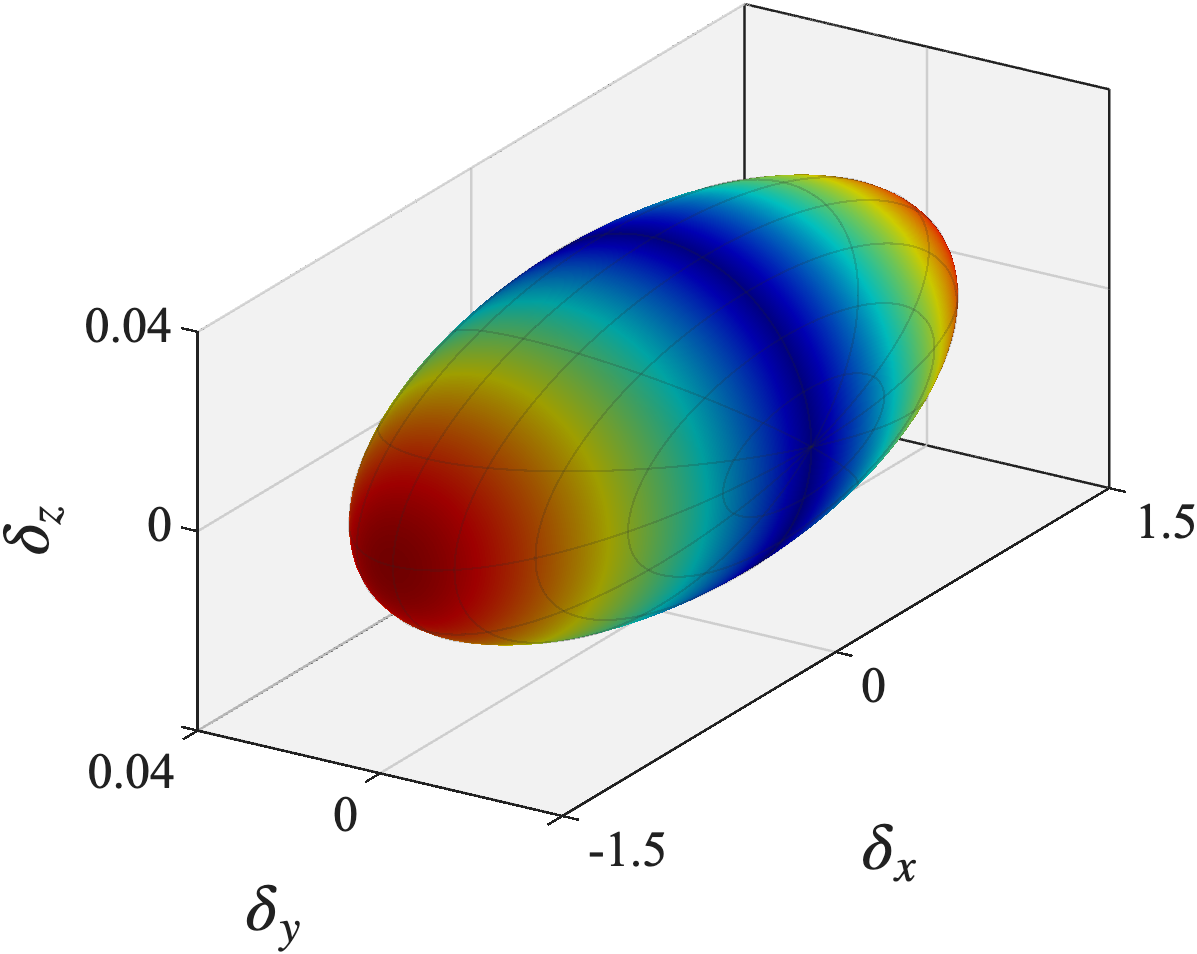}
        \caption{\(D=6\)~m.}
    \end{subfigure}
    \caption{3D confusability ellipsoids at the different distances.}
    \label{fig:ellipsoid_distance_evolution}
\end{figure}

Fig.~\ref{fig:ellipsoid_distance_evolution} visualizes the confusability ellipsoid in \eqref{eq:local_confusability_ellipsoid}, where the ellipsoid is evaluated at the  anchor \(\mathbf r=\mathbf 0\) for \(D=3\)~m and \(D=6\)~m. We can see that the principal axis along the range direction is much longer than the two transverse axes, which indicates that range displacements are much harder to resolve and agrees with the anisotropy revealed in Figs.~\ref{fig:range_displacement_reliability} and \ref{fig:transverse_displacement_reliability}. Moreover, as the controllable region moves away from the array, the range half-length grows much more dramatically than the transverse ones, which indicates a progressive transition toward the far-field region.

To further illustrate the near-to-far-field transition, Fig.~\ref{fig:range_3d} plots the range Bhattacharyya distance \(B(\mathbf r_u,\mathbf r_u+\mathbf e_x)\) versus the distance \(D\) and the anchor offset \(u\), where \(\mathbf r_u=(0,u,u)^{\mathrm T}\) is the diagonal anchor. It is observed that the Bhattacharyya distance decays rapidly with \(D\) and falls below \(B_{\mathrm{req}}\) when \(D\) is sufficiently large, which indicates that range displacements are no longer resolvable and the 3D spatial sensing degenerates to the 2D sensing over the transverse plane, as described by Proposition~\ref{prop:far_field_local_metric}.

\begin{figure}[t]
    \centering
    \includegraphics[width=0.9\linewidth]{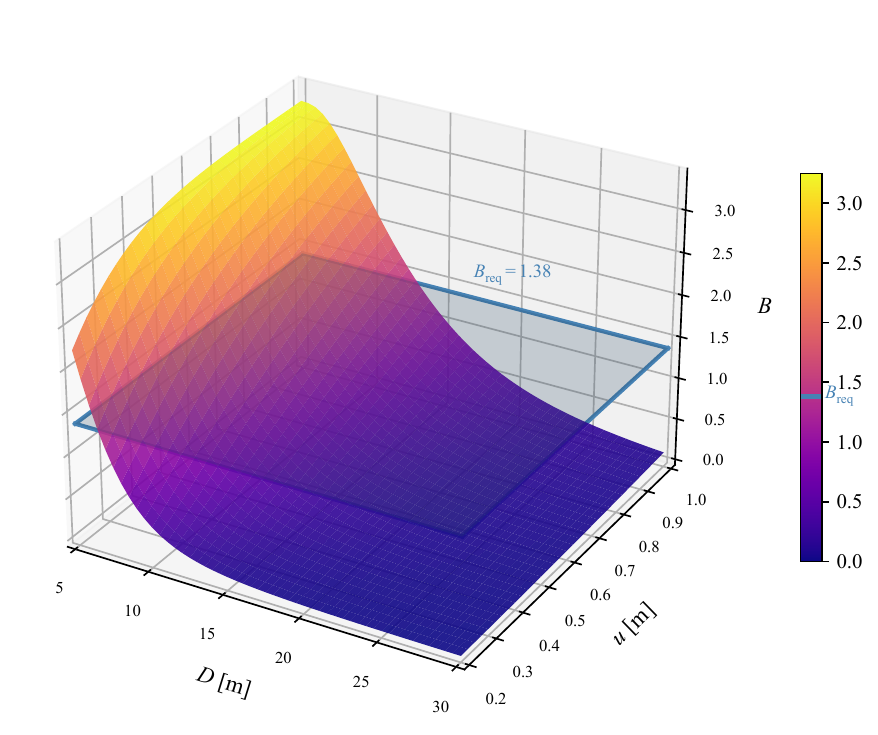}
    \caption{Range Bhattacharyya distance  versus the distance \(D\) and the anchor offset \(u\).}
    \label{fig:range_3d}
\end{figure}

\subsection{Comparison with Different Lattice Codebooks}

\begin{figure}[t]
    \centering
    \includegraphics[width=0.9\linewidth]{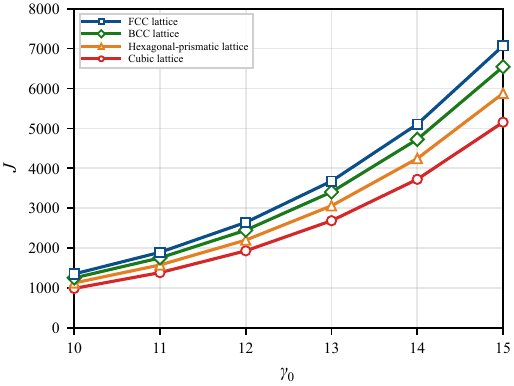}
    \caption{Codebook size \(J\) versus the sensing SNR \(\gamma_0\) for different lattice codebooks.}
    \label{fig:codebook_design_comparison}
\end{figure}

Fig.~\ref{fig:codebook_design_comparison} compares the achievable codebook size \(J\) of the proposed FCC codebook with three baselines versus the sensing SNR, where \(L=30\). Specifically, the BCC lattice is formed by adding a body-centered point to the cubic lattice, the hexagonal-prismatic lattice is generated by stacking hexagonal layers to form hexagonal prisms, and the cubic lattice places codewords directly on the vertices of the cubic. All four schemes share the same volume \(\mathcal V_G\) in \eqref{eq:metric_volume} and the same minimum-distance requirement \(d_J\) in \eqref{eq:fcc_minimum_distance}, and differ only in their fundamental-cell volumes. It can be seen the proposed FCC codebook achieves the largest alphabet over the entire SNR range, and its alphabet is approximately 37\% larger than that of the cubic packing. This is because a smaller fundamental cell permits more codewords within the same metric volume, which validates the FCC codebook design in Section~\ref{sec:codebook_design}.

\subsection{Near-field Capacity Gain}

\begin{figure}[t]
    \centering    \includegraphics[width=0.9\linewidth]{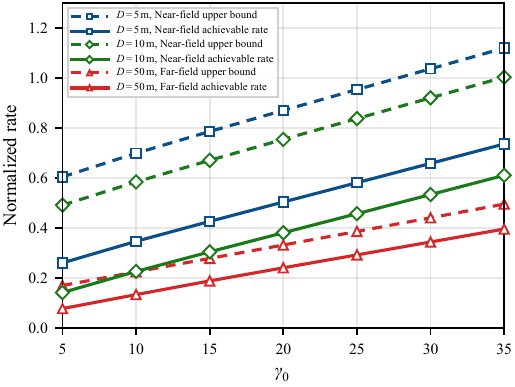}
    \caption{Achievable rate and converse bound versus the sensing SNR \(\gamma_0\) at different distances.}
    \label{fig:capacity_vs_snr}
\end{figure}

\begin{figure}[t]
    \centering
\includegraphics[width=0.9\linewidth]{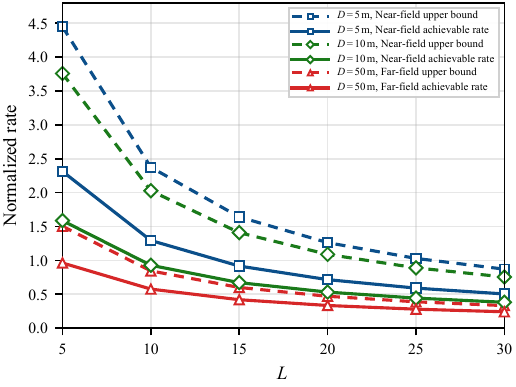}
    \caption{Achievable rate and converse bound versus the number of sensing snapshots \(L\) at different distances.}
    \label{fig:capacity_vs_snapshots}
\end{figure}

Fig.~\ref{fig:capacity_vs_snr} compares the achievable rate and the converse bound versus the sensing SNR \(\gamma_0\) for \(D\in\{5,10,50\}\)~m with \(\xi=0.3\) and \(L=30\). Since \(D=50\)~m lies beyond the Rayleigh distance, the corresponding rates are evaluated with the 2D plane packing. It is observed that both the achievable rate and the converse bound decrease monotonically as \(D\) increases over the entire SNR range, which indicates that a closer controllable region enjoys finer near-field resolvability and therefore a higher embodied information capacity.

Fig.~\ref{fig:capacity_vs_snapshots} further explores the effect of the number of sensing snapshots, where the achievable rate and the converse bound are plotted versus \(L\) for \(D\in\{5,10,50\}\)~m with \(\gamma_0=20\)~dB and \(\xi=0.3\). It is observed that the normalized rate decreases monotonically with \(L\) at all distances, since the total information \(\log_2J\) increases with \(L\) only sublinearly, which reveals that more snapshots support a larger embodied alphabet but yield less information per probing duration. Moreover, the rates at a closer distance remain higher for all \(L\), which indicates that the near-field capacity gain persists regardless of the number of sensing snapshots, as long as the controllable region is within the near-field region.

\section{Conclusion}
\label{sec:conclusion}

This paper developed a finite-snapshot theory of near-field embodied communication by connecting sensing-induced statistical distinguishability with the 3D spatial packing of physical symbols.

The main insight is that sensing resolution can be viewed as a communication resource. In the near field, this resource is genuinely volumetric but nonuniform: the communication value of a controllable region is determined not only by its physical size, but also by its location relative to the array and by the directional resolution available within it. Embodied communication systems should therefore be designed around the spatial distribution of sensing resolution, rather than Euclidean separation alone. The transition to the far field further validates that range resolution is the critical resource that separates volumetric embodied alphabets from planar ones.

More broadly, this work connects the evolution of sensing-capable wireless infrastructure with the emergence of embodied AI. As intelligent agents gain the ability to act on the physical world, their motions, configurations, and modifications of the environment can become information-bearing degrees of freedom. When these deliberately selected physical states are observed by sensing-capable infrastructure, communication can be embedded directly into physical interaction rather than treated solely as a separate waveform-level process. From this perspective, the environment becomes a shared interface linking perception, action, and communication, pointing toward intelligent systems in which embodied agents communicate not only by transmitting signals, but also by purposefully shaping the physical world.

\appendices

\section{Proof of Lemma~\ref{lem:bhattacharyya_bound}}
\label{app:bhattacharyya_bound}
\begin{proof}
For the given physical state \(\mathbf r_i\), the ML detector errs whenever \(p(\mathbf Y\mid\mathbf r_j)\ge p(\mathbf Y\mid\mathbf r_i)\), and the error region is given by
\begin{align*}
P_{i\rightarrow j}
&=
\int_{\{\mathbf Y:\,p(\mathbf Y\mid\mathbf r_j)\ge p(\mathbf Y\mid\mathbf r_i)\}}
p(\mathbf Y\mid\mathbf r_i)\,d\mathbf Y \nonumber\\
&\le
\int_{\{\mathbf Y:\,p(\mathbf Y\mid\mathbf r_j)\ge p(\mathbf Y\mid\mathbf r_i)\}}
\sqrt{p(\mathbf Y\mid\mathbf r_i)\,p(\mathbf Y\mid\mathbf r_j)}
\,d\mathbf Y \nonumber\\
&\le
\int
\sqrt{p(\mathbf Y\mid\mathbf r_i)\,p(\mathbf Y\mid\mathbf r_j)}
\,d\mathbf Y,
\end{align*}
where the first inequality follows from \(p(\mathbf Y\mid\mathbf r_i)\le\sqrt{p(\mathbf Y\mid\mathbf r_i)\,p(\mathbf Y\mid\mathbf r_j)}\) on the error region, and the last inequality follows from extending the integral to the full space. Since the \(L\) snapshots are i.i.d., the joint density factorizes as \(p(\mathbf Y\mid\mathbf r_k)=\prod_{\ell=1}^{L}p(\mathbf y_\ell\mid\mathbf r_k)\) for \(k\in\{i,j\}\), and hence
\begin{align*}
\int\!\!\!
\sqrt{p(\mathbf Y\mid\mathbf r_i)\,p(\mathbf Y\mid\mathbf r_j)}
\,d\mathbf Y
&\!=\!
\prod_{\ell=1}^{L}\!\!
\int_{\C^{M_{\mathrm v}}}\!\!\!\!
\sqrt{p(\mathbf y_\ell\mid\mathbf r_i)\,p(\mathbf y_\ell\mid\mathbf r_j)}
\,d\mathbf y_\ell
\nonumber\\
&=
\prod_{\ell=1}^{L}\operatorname{BC}(\mathbf r_i,\mathbf r_j)
=
\operatorname{BC}(\mathbf r_i,\mathbf r_j)^{L},
\end{align*}
where the last two equalities follow from the Bhattacharyya coefficient given by \eqref{eq:bhattacharyya_coefficient}, which is identical for every snapshot. This yields \eqref{eq:bhattacharyya_bound}.
\end{proof}

\section{Proof of Lemma~\ref{lem:gaussian_bhattacharyya}}
\label{app:gaussian_bhattacharyya}
\begin{proof}
For \(k\in\{i,j\}\), the complex Gaussian density under hypothesis \(\mathcal H_k\) in \eqref{eq:binary_hypotheses} is given by
\begin{align*}
p(\mathbf y\mid\mathbf r_k)
=
\frac{1}{\pi^{M_{\mathrm v}}\det(\mathbf R_k)}
\exp\!\left(
-\mathbf y^{\mathrm H}\mathbf R_k^{-1}\mathbf y
\right).
\end{align*}
Therefore,
\begin{align}
\sqrt{p(\mathbf y\mid\mathbf r_i)p(\mathbf y\mid\mathbf r_j)}
&=
\frac{
\exp\!\left[
-\frac{1}{2}\mathbf y^{\mathrm H}
(\mathbf R_i^{-1}+\mathbf R_j^{-1})
\mathbf y
\right]
}{
\pi^{M_{\mathrm v}}\sqrt{\det(\mathbf R_i)\det(\mathbf R_j)}
}.
\label{eq:gaussian_density}
\end{align}
Define \(\mathbf A_{ij}\triangleq(\mathbf R_i^{-1}+\mathbf R_j^{-1})/2\). Then, based on the  \eqref{eq:bhattacharyya_coefficient} and \eqref{eq:gaussian_density}, we have
\begin{align}
\operatorname{BC}_{ij}
&=
\frac{1}{\sqrt{\det(\mathbf R_i)\det(\mathbf R_j)}}
\int_{\C^{M_{\mathrm v}}}
\pi^{-M_{\mathrm v}}
\exp(-\mathbf y^{\mathrm H}\mathbf A_{ij}\mathbf y)
d\mathbf y \nonumber\\
&=
\frac{\det(\mathbf A_{ij})^{-1}}
{\sqrt{\det(\mathbf R_i)\det(\mathbf R_j)}},
\label{eq:gaussian_density_bc}
\end{align}
where the last equality follows from the standard complex Gaussian integral
\[
\int_{\C^{M_{\mathrm v}}}
\pi^{-M_{\mathrm v}}
\exp(-\mathbf y^{\mathrm H}\mathbf A_{ij}\mathbf y)
d\mathbf y
=
\det(\mathbf A_{ij})^{-1}.
\]
Moreover, since  \(\mathbf R_i^{-1}+\mathbf R_j^{-1}=\mathbf R_i^{-1}(\mathbf R_i+\mathbf R_j)\mathbf R_j^{-1}\), we have
\begin{align}
\det(\mathbf A_{ij})\!
=\!
\det\!\left(\frac{\mathbf R_i^{-1}+\mathbf R_j^{-1}}{2}\right)\! =\!
\frac{
\det\!\left((\mathbf R_i+\mathbf R_j)/2\right)
}{
\det(\mathbf R_i)\det(\mathbf R_j)
},
\label{eq:gaussian_density_det}
\end{align}
Substituting \eqref{eq:gaussian_density_det} into \eqref{eq:gaussian_density_bc} gives
\begin{align*}
\operatorname{BC}_{ij}
=
\frac{
\sqrt{\det(\mathbf R_i)\det(\mathbf R_j)}
}{
\det\!\left((\mathbf R_i+\mathbf R_j)/2\right)
}.
\end{align*}
Taking \(B(\mathbf r_i,\mathbf r_j)=-\log\operatorname{BC}_{ij}\) yields \eqref{eq:bhattacharyya_general}.
\end{proof}

\section{Proof of Theorem~\ref{thm:pairwise_reliability}}
\label{app:pairwise_reliability}
\begin{proof}
Applying Sylvester's determinant identity to the covariance matrix \(\mathbf R_i\) in \eqref{eq:state_covariance} gives
\begin{align}
\det(\mathbf R_i)
=
\sigma^{2M_{\mathrm v}}
\big(1+\gamma_0\mathbf a^{\mathrm H}(\mathbf r_i)\mathbf a(\mathbf r_i)\big)
\!=\!
\sigma^{2M_{\mathrm v}}(1+\gamma_0),
\label{eq:app_det_single_cov}
\end{align}
where the last equality follows from \(\|\mathbf a(\mathbf r_i)\|_2=1\), and the same expression holds for \(\det(\mathbf R_j)\).

Define \(\mathbf U_{ij}=[\mathbf a(\mathbf r_i),\mathbf a(\mathbf r_j)]\in\C^{M_{\mathrm v}\times 2}\), and then the averaged covariance can be expressed as
\begin{align*}
\frac{\mathbf R_i+\mathbf R_j}{2}
=
\sigma^2
\left(
\mathbf I_{M_{\mathrm v}}+\frac{\gamma_0}{2}\mathbf U_{ij}\mathbf U_{ij}^{\mathrm H}
\right).
\end{align*}
Applying Sylvester's determinant identity again gives
\begin{align}
\det\!\left(\frac{\mathbf R_i+\mathbf R_j}{2}\right)
&=
\sigma^{2M_{\mathrm v}}
\det\!\left(
\mathbf I_2+\frac{\gamma_0}{2}\mathbf U_{ij}^{\mathrm H}\mathbf U_{ij}
\right)
\nonumber\\
&=
\sigma^{2M_{\mathrm v}}
\left[
\left(1+\frac{\gamma_0}{2}\right)^2
-
\frac{\gamma_0^2}{4}\,\eta(\mathbf r_i,\mathbf r_j)
\right],
\label{eq:app_average_det}
\end{align}
where the second equality derives from
\begin{align*}
\mathbf I_2+\frac{\gamma_0}{2}\mathbf U_{ij}^{\mathrm H}\mathbf U_{ij}
=
\begin{bmatrix}
1+\frac{\gamma_0}{2} & \frac{\gamma_0}{2}\mathbf a^{\mathrm H}(\mathbf r_i)\mathbf a(\mathbf r_j)\\
\frac{\gamma_0}{2}\mathbf a^{\mathrm H}(\mathbf r_j)\mathbf a(\mathbf r_i) & 1+\frac{\gamma_0}{2}
\end{bmatrix}
\end{align*}
and the definition of \(\eta(\mathbf r_i,\mathbf r_j)=|\mathbf a^{\mathrm H}(\mathbf r_i)\mathbf a(\mathbf r_j)|^2\).

Finally, substituting \eqref{eq:app_det_single_cov} and \eqref{eq:app_average_det} into \eqref{eq:bhattacharyya_general}, we have
\begin{align*}
\frac{
\left(1+\frac{\gamma_0}{2}\right)^2-\frac{\gamma_0^2}{4}\eta
}{
1+\gamma_0
}
=
1+\frac{\gamma_0^2}{4(1+\gamma_0)}(1-\eta)
=
1+\kappa(1-\eta),
\end{align*}
where \(\kappa=\gamma_0^2/(4(1+\gamma_0))\). This yields \eqref{eq:pairwise_bhattacharyya}. Further substituting \eqref{eq:pairwise_bhattacharyya} into \eqref{eq:pairwise_error_bound} gives
\begin{align*}
P_{i\rightarrow j}
\le
\exp\!\big[-L B(\mathbf r_i,\mathbf r_j)\big]
=
\Big(
1+\kappa\big[1-\eta(\mathbf r_i,\mathbf r_j)\big]
\Big)^{-L},
\end{align*}
which is exactly \eqref{eq:pairwise_error_bound_closed}.
\end{proof}

\bibliographystyle{IEEEtran}
\bibliography{IEEEabrv,references}

\end{document}